\documentclass[12pt]{article}

\usepackage{mathtools}
\usepackage{booktabs}
\usepackage[english]{babel} 
\usepackage[protrusion=true,expansion=true]{microtype} 
\usepackage{amsmath,amsfonts,amsthm}
\usepackage{ amssymb }
\usepackage{graphicx}
\usepackage{array}
\usepackage{bbm}
\usepackage{multirow}
\usepackage{hhline}
\usepackage[titletoc]{appendix}

\usepackage{xr}
\usepackage{booktabs} 
\usepackage{natbib}
\usepackage{setspace}
\usepackage{microtype} 
\usepackage{subcaption}
\usepackage{tabularx}
\usepackage[font = small,labelfont=bf,textfont=it]{caption} 
\usepackage{footnote}
\usepackage{algorithm}
\usepackage[noend]{algpseudocode}

\makeatletter
\newcommand{\distas}[1]{\mathbin{\overset{#1}{\kern\z@\sim}}}%
\newsavebox{\mybox}\newsavebox{\mysim}
\newtheorem{theorem}{Theorem}[section]
\newtheorem{corollary}{Corollary}[section]

\theoremstyle{definition}

\newtheorem{theo}{Theorem}[subsection]
\newtheorem{prop}{Proposition}[subsection]
\newtheorem{assum}{Assumption}[subsection]
\newtheorem{defn}{Definition}[subsection]
\newtheorem{lemmmm}{Lemma}[subsection]
\newtheorem{coroll}{Corollary}[subsection]

\newtheorem{remark}[theorem]{Remark}
\newcommand{\distras}[1]{%
  \savebox{\mybox}{\hbox{\kern3pt$\scriptstyle#1$\kern3pt}}%
  \savebox{\mysim}{\hbox{$\sim$}}%
  \mathbin{\overset{#1}{\kern\z@\resizebox{\wd\mybox}{\ht\mysim}{$\sim$}}}%
}
\newcolumntype{C}[1]{>{\centering\let\newline\\\arraybackslash\hspace{0pt}}m{#1}}

\newcommand{\blind}{1}

\usepackage{ragged2e}
\begin{document}

\def\spacingset#1{\renewcommand{\baselinestretch}%
{#1}\small\normalsize} \spacingset{1}


\if1\blind
{
  \centering{\bf\Large Multi-Resolution Functional ANOVA for Large-Scale, Many-Input Computer Experiments}\\
  \vspace{0.2in}
  \centering{Chih-Li Sung$^{*,a}$, Wenjia Wang\footnote{These authors contributed equally to the manuscript.}$^{,b}$, Matthew Plumlee$^c$, Benjamin Haaland\footnote{
    The authors gratefully acknowledge funding from NSF DMS-1739097 and DMS-1564438.}$^{,d,e}$\\
      \vspace{0.2in}
  \centering{
  $^a$Michigan State University\\
  $^b$Statistical and Applied Mathematical Sciences Institute\\
    $^c$Northwestern University\\
    $^d$University of Utah\\
    $^e$Georgia Institute of Technology}
} \fi

\if0\blind
{
  \bigskip
  \bigskip
  \bigskip
  \begin{center}
    {\bf\Large Multi-Resolution Functional ANOVA for Large-Scale, Many-Input Computer Experiments}
\end{center}
  \medskip
} \fi
\begin{abstract}
The Gaussian process is a standard tool for building emulators for both deterministic and stochastic computer experiments. However, application of Gaussian process models is greatly limited in practice, particularly for large-scale and many-input computer experiments that have become typical. We propose a multi-resolution functional ANOVA model as a computationally feasible emulation alternative. More generally, this model can be used for large-scale and many-input non-linear regression problems. 

An overlapping group lasso approach is used for estimation, ensuring computational feasibility in a large-scale and many-input setting. 
New results on consistency and inference
for the (potentially overlapping) group lasso in a high-dimensional setting are developed and applied to the proposed multi-resolution functional ANOVA model.
Importantly, these results 
allow us to quantify the uncertainty in our predictions.

Numerical examples demonstrate that the proposed model enjoys marked computational advantages. Data capabilities, both in terms of sample size and dimension, meet or exceed best available emulation tools while meeting or exceeding emulation accuracy. 
\end{abstract}

\justify

\noindent%
\vfill

\newpage
\spacingset{1.45}

\section{Introduction}

Computer models are implementations of
complex mathematical models using computer
codes.  They are used to study systems
of interest for which physical experimentation
is either infeasible or very limited.
For example, \cite{hotzer2015large} model crystalline micro-structure of alloys as a function of solidification velocity.
Another example is the simulation of 
population-wide cardiovascular effects based on
salt intake in the U.S. presented in \cite{bibbins2010projected}.

Calibration, exploration, and optimization of a computer model requires the response given many potential inputs. Computer models are often too computationally demanding for free generation of input/response combinations. A well-established  solution to this problem is the use of \emph{emulators} \citep{sacks1989design,santner2013design}.   This solution involves evaluating the  response at a series of well-distributed inputs.
Then, 
an emulator of the computer model is built using the collected data. Calibration, exploration, or optimization can then be carried out on
the emulator directly
\citep{pratola2016bayesian,santner2013design,goh2013prediction,wang2013optimisation,
asmussen2007stochastic,fang2005design}.

A standard method for building emulators after deterministic or stochastic computer experiments
is Gaussian process \citep{santner2013design}, or almost equivalently \citep{lukic2001stochastic} reproducing kernel Hilbert space regression \citep{wahba1990spline} . 
Gaussian process modeling leverages known properties of the underlying response surface to produce mathematically
simple predictions and statistical uncertainty quantification via confidence intervals after an experiment. 

Unfortunately, the use of Gaussian process emulators is limited for large-scale computer experiments.    
Let $X=\{x_1,\ldots,x_n\}$ denote the set of input locations for the experiment, $f(x)$ the computer model response at input $x$, and $\Phi(x,x')$ the kernel function at inputs $x$ and $x'$. 
Further, let $\Phi(X,X)$ denote the $n \times n$ matrix with entries $\Phi(x_i,x_j)$ and $f(X)$ the length $n$ vector of responses $f(x_i)$.  
The simplest form of Gaussian process emulator is then found by solving for the $n$ vector $\alpha$ with $\Phi(X,X) \alpha= f(X)$.   There are at least three major challenges that prevent using the Gaussian process emulator as $n$ gets large, ranked roughly in order of consequence for typical combinations of sample size, kernel, and experimental design.
\emph{(i)} More than $n^2/2$  values are needed to represent
$\Phi(X,X)$, which 
can cause memory challenges, particularly on a personal computer and for a non-sparse $\Phi(X,X)$.
\emph{(ii)} Numeric solutions to  $\Phi(X,X) \alpha= f(X)$ 
can be highly unstable, so that more data can lead to less accurate results.
\emph{(iii)} The computational complexity for solving
the linear system $\Phi(X,X) \alpha = f(X)$ can be burdensome for large $n$.

Overcoming these problems, which are also key bottlenecks for many related statistical methods, is an active area of research, particularly in statistical emulation of computer experiments. 
While much progress has been made in this area, much work remains. 
There have been partial solutions proposed in the literature: using less smooth kernels can address \emph{(ii)} \citep{wendland2004scattered}, covariance tapering \emph{(i, ii)} \citep{furrer2006covariance,kaufman2011efficient},  a nugget effect \emph{(ii)} \citep{ranjan2011computationally}, multi-step emulators \emph{(i, ii)} \citep{haaland2011accurate}, specialized design \emph{(i, iii)} \citep{plumlee2014fast}, and parallelization and computational methods \emph{(ii)} \citep{paciorek2013parallelizing}.
To address all three challenges simultaneously, one must exploit features present in the response surface.  Local approaches to emulation address \emph{(i, ii, iii)} using the principle that only a fraction of the total responses from an experiment  are needed to achieve accurate prediction at a particular input of interest \citep{sung2016potentially,gramacy2015speeding,gramacy2015local,gramacy2014massively}. 

This article discusses a new multi-resolution functional ANOVA (MRFA) approach to emulation of large-scale (large $n$) and many-input (many-dimensional $x$) computer experiments.  
The MRFA operates by exploiting features which are commonly encountered in practical computer models. 
The remainder of this article is organized as follows.
In Section \ref{MRFA}, we provide background and preliminary results, then introduce the MRFA model.
In Section \ref{fitting}, we formulate the model fitting as an overlapping group lasso problem and discuss efficient model fitting, as well as tuning parameter selection.
In Section \ref{inference}, we present new results on consistency in the presence of approximation bias.
In Section \ref{CIs}, we present new results on large-sample hypothesis testing for the high-dimensional, potentially overlapping, group lasso problem in the stochastic case.
A heuristic approach, with coverage correction, is presented for the deterministic case.
The tests are then inverted to obtain pointwise confidence intervals on the regression function.
Basis function selection is discussed in Section \ref{basisFunctions}.
In Section \ref{examples}, we present a few illustrative examples showcasing the capabilities of the MRFA technique in a large-scale, many-input setting.
Finally, in Section \ref{discussion}, we close with a brief discussion.
Proofs are provided in the Appendix.

\section{Multi-Resolution Functional ANOVA}\label{MRFA}

The motivation for the multi-resolution functional ANOVA emulator is as follows.  First, note that a function with a low-dimensional input can easily be approximated given a large number of responses provided sufficient smoothness.  One does not have to use anything as complex as even the simplest Gaussian process regression to achieve good emulation, and in many cases Gaussian process regression would fail for the reasons discussed in the introduction. For example, if one has $n=100,000$, then a Gaussian process emulator has $100,000$ basis functions, which is far more than necessary for arbitrarily high-accuracy approximation of most low-dimensional functions. Consider the example shown in Figure \ref{fig:BasisRepresentation}.   
In the example, 1000 evenly spaced data points are collected. 
Using Wendland's kernel \citep{wendland1995piecewise} with $k=4$ and width $0.75$ implies $\Phi(X,X)$ has condition number $4.6\times 10^{22}$,
so that the matrix inverse is not useful in a floating point setting. 
Briefly, Wendland's kernels are compactly supported kernels expressed as truncated polynomials, with $k=4$ and width $0.75$ ensuring that the kernels have $2k=8$ continuous derivatives with non-zero support radius $0.75$.
More detail on Wendland's kernels is provided in Section \ref{basisFunctions}.
Back to the function approximation example problem, we see that the true function is reasonably well-approximated by the set of five basis functions shown in gray in the left panel and very well-approximated by the set of 15 basis functions shown in gray in the right panel.  
This type of multi-resolution emulation \citep{nychka2014multi} has been successfully employed for function approximation, particularly in a low-dimensional input setting.

\begin{figure}[!h]
\centering
\includegraphics[width=\textwidth]{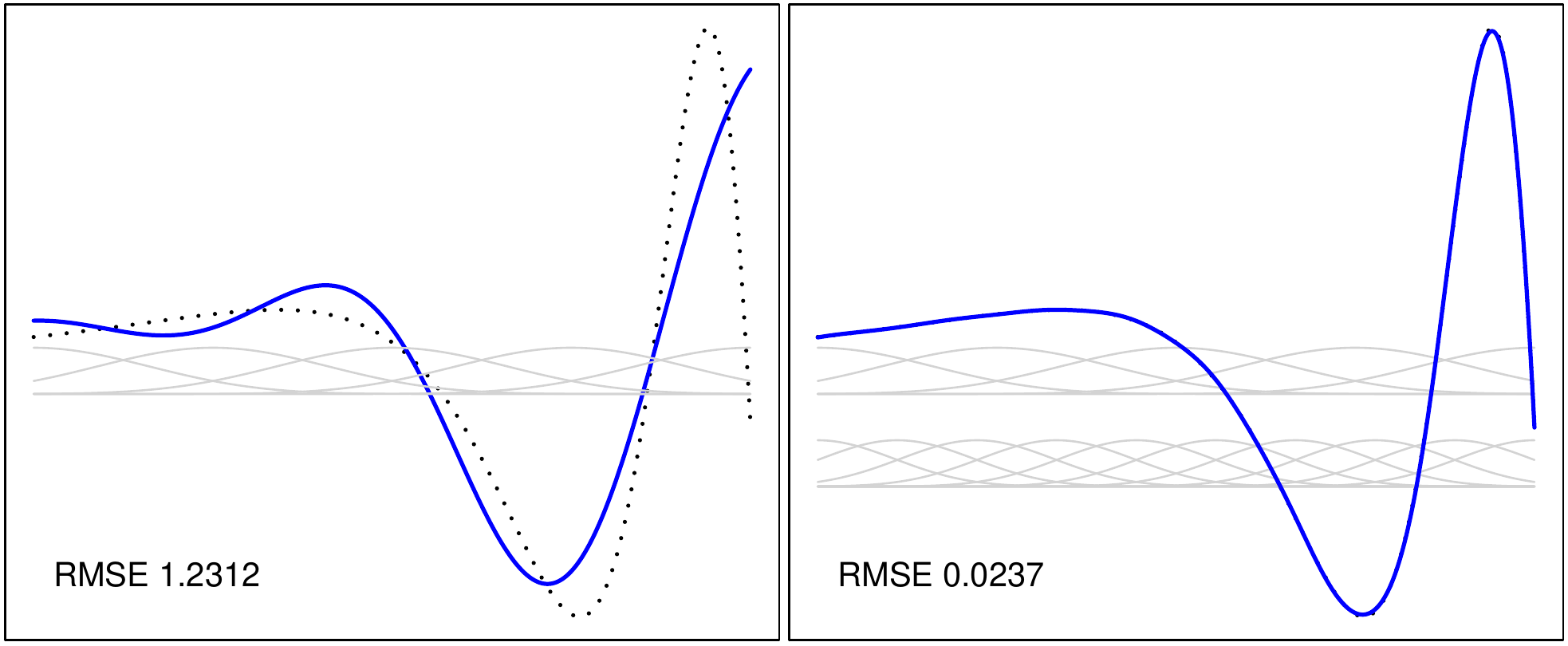}
\caption{Multi-resolution example with 5 basis function (left panel) and 15 basis functions (right panel). Here, the true function is shown in dotted black, the emulator in solid blue, and the basis functions are Wendland's kernels with $k=4$ and widths $0.75$ and $0.50$, shown in solid light gray.}
\label{fig:BasisRepresentation}
\end{figure}

Approximating easily in low-dimensions does not directly improve approximations in higher-dimensions, where coming up with a good set of basis functions is an onerous task.
Roughly speaking, if an unknown function has a high-dimensional input and no simplifying structure, then the exercise of trying to build an accurate emulator with finite data is essentially hopeless, so a means for detecting simplifying structure should be a corner-stone of any proposed technique.

Consider a relatively low-order functional ANOVA, where a function is represented as a sum of main effect functions, two-way interaction functions and so on.  Functional ANOVA has played an important role in variable screening  for many-input computer experiments. 
See for example Chap. 6.3 of \cite{fang2005design} or Chap. 7.1 of \cite{santner2013design}. 
Functional ANOVA has also been used for function approximation across a spectrum of other applications. 
For example, \cite{owen1997monte} used a functional ANOVA representation to approximate the variance of scrambled net quadrature and \cite{stone1997polynomial} approximated a general regression function using a functional ANOVA structure. 
By considering a function with a low-order functional ANOVA, the curse of dimensionality can be largely sidestepped.
While this modeling approach can increase the flexibility of additive modeling, it retains much of the interpretability. 

Our proposed multi-resolution functional ANOVA approach respects two types of \textit{strong} effect heredity \citep{wu2011experiments}, 
{\it (i)} in the order of functional ANOVA,
so that higher-order interaction functions are only entertained if all their lower-dimensional components are present, 
and 
{\it (ii)} in the {resolution} of approximation to these relatively low-dimensional component functions, so that not too many basis functions are used.
The hope is that by targeting a simpler representation (low-order functional ANOVA model), which is amenable to low-dimensional approximation (via multi-resolution model), accurate emulators can be formed in a very large-scale and many-input setting.

For an integrable function $f:\Omega\to\mathbb{R},\;\Omega\subset\mathbb{R}^d$, a functional ANOVA can be defined recursively as follows.
Let
$f_{\emptyset} =  \int_{\Omega} f(x) {\rm d}x$
and 
\begin{gather}
f_u(x) = \int_{\Omega_{-u}} \left(f(x) - \sum_{v \subsetneq u} f_v (x) \right) {\rm d}x_{-u}.\label{FA}
\end{gather}
Here, $u,v\subset\mathcal{D}=\{1,\ldots,d\}$ denote sets of indices
and
the notation $\int_{\Omega_{-u}} \cdots {\rm d}x_{-u}$ indicates integration over the variables not in $u$ for a fixed value of $x_u$. 
%
%
Now, 
$f$ can be represented via its ANOVA decomposition as
\begin{gather}
f(x) = \sum_{u \subseteq \mathcal{D}} f_u(x).\nonumber
\end{gather}
%
%
%
Note that in this decomposition, each component function $f_u(x)$ is a function of $x$ that only depends on $x_u$. 
$f_{\emptyset}$ is often referred to as the \textit{mean} function, 
$f_{\{i\}}(x)$, $i\in\mathcal{D}$ as the \textit{main effect} functions, 
$f_{\{i,j\}}(x)$, $i,j\in\mathcal{D},i\ne j$ as the \textit{two-way interaction} functions,
and so on.
The terms in the functional ANOVA (\ref{FA}) are orthogonal in $L_2(\Omega)$, which ensures uniqueness of the representation.
Generally, there is no closed form for the component functions $f_u$, so Monte Carlo techniques are commonly used to approximate them.

It turns out that if the full-dimensional function $f$ lives in a reproducing kernel Hilbert space (RKHS) \citep{aronszajn1950theory} on $[0,1]^d$ with a {product kernel}, then 
$f$ can be represented as a sum of 
component functions $f_u$, which live in RKHS's
whose kernels (and therefore norms) are determined by the full-dimensional kernel.
This result is summarized in Theorem \ref{thm:basis}, whose proof is given in Appendix \ref{append:basis}.
Define an RKHS $\mathcal{N}_\Phi(\Omega)$ for a symmetric positive-definite kernel $\Phi:\Omega\times\Omega\rightarrow\mathbb{R}$ as the \emph{closure} of the normed linear space,
\[\left\{ \left. \sum_{x \in X} \beta_x \Phi(\cdot, x) \right| \beta_x \in \mathbb{R},\;x\in\Omega \right\}, \]
with inner product $\sum_{x \in X} \sum_{y \in Y}  \alpha_x \beta_y \Phi(x, y)$ for component functions $\sum_{x \in X} \alpha_x \Phi(\cdot,x)$ and $\sum_{y \in Y} \beta_y \Phi(\cdot,y)$. 

\begin{theorem}\label{thm:basis}
Suppose $\Phi\in\Omega\times\Omega\rightarrow\mathbb{R}$ is a symmetric positive-definite kernel on $\Omega=[0,1]^d$ and $\Phi$ has a product structure, 
$\Phi(x,y)=\prod_{j=1}^d\phi_j(x_j,y_j)$.
Then, any $f\in\mathcal{N}_\Phi([0,1]^d)$
has representation $f=\sum_{u\subseteq\mathcal{D}}f_u$, where
$f_u\in\mathcal{N}_{\Phi_u}([0,1]^{|u|})$ and
$\Phi_{u}=\prod_{j \in u} \phi_j$,
where $|A|$ denotes the cardinality of a set $A$.
\end{theorem}

The proposed emulator is a low-resolution representation of a low-order functional ANOVA, $\hat{f}_{\rm ANOVA}$.
Clearly, this process introduces approximation errors due to both the resolution and the order of the ANOVA.
On the other hand, it is anticipated that for target functions encountered in practice, inaccuracy due to the low-order functional ANOVA and low-resolution approximation will be small. In other words, high-order interaction functions will be negligible and low-dimensional component functions will be well-approximated by a relatively small set of basis functions.

An MRFA emulator can be represented as
\begin{gather}
\hat{f}_{\rm MRFA}(x)=\sum_{u\in\mathcal{E}}\sum_{r\le R(u)}\hat f_{u,r}(x),\nonumber
\end{gather}
where 
$\mathcal{E}$ is a {set} of sets of indices which obeys strong effect heredity (if a set of indices is in $\mathcal{E}$, then every one of its subsets is also in $\mathcal{E}$) and 
$R(u)\in\mathbb{N}$ denotes the resolution {level} used to represent component function $f_u$.  
If each $\hat f_{u,r}$ is represented as a linear combination of $n_{u}(r)$ basis functions $\varphi_u^{rk}:\mathbb{R}^{|u|}\to\mathbb{R}$, $k=1,\ldots,n_{u}(r)$, 
then
\begin{gather}
\hat{f}_{\rm MRFA}(x)=\sum_{u\in\mathcal{E}}\sum_{r\le R(u)}\sum_{k=1}^{n_u(r)}\hat{\beta}_u^{rk}\varphi_u^{rk}(x_u).\nonumber
\end{gather}
For simplicity, the {level} of resolution is taken in pre-specified increments indexed by 
positive integers. 
$\mathcal{E}$ could also conceivably be a set of sets of indices which obeys \textit{weak} effect heredity (if a set of indices is in $\mathcal{E}$, then \textit{at least} one of its subsets of size one smaller
is also in $\mathcal{E}$). Depending on the objectives of the studies, either strong or weak effect heredity could be considered and the development herein is unchanged. 
On the other hand, strong effect heredity has computational advantages because more models are ruled out from the model search, while weak effect heredity may become computationally prohibitive in a many-input setting.   

It is important to note that for the proposed multi-resolution functional ANOVA model, we do not require zero means or orthogonality of components functions. 
While 
these properties 
ensure identifiability in a standard functional ANOVA model, as in equation (\ref{FA}), they are not required for obtaining an accurate representation.
A setup of the multi-resolution functional ANOVA which does satisfy mean zero, orthogonal effect functions could be obtained in a straightforward manner by forming functional ANOVA representations of the basis functions selected based on resolution and smoothness concerns (as outlined in Section \ref{basisFunctions}), then grouping terms appropriately. 
We chose not to pursue this line of development here because our primary interest is in strong effect heredity as a mechanism for encouraging simplicity of the function approximation.
Additionally, interpretability for the proposed multi-resolution functional ANOVA model and a standard functional ANOVA representation is similar, given the challenge of interpreting interaction functions outside the context of their parent effect functions.

{The proposed MRFA model is an example of a many-dimensional nonparametric regression model. 
In the surrounding literature, a large body of work has focused on additive models with main effect functions, such as generalized additive models (GAM) \citep{hastie1990}, regularization of derivative expectation operator (RODEO) \citep{wasserman2005rodeo} and sparse additive models (SpAM) \citep{ravikumar2009spam}. 
Related work has applied a functional ANOVA perspective to additive models, such as multivariate adaptive regression splines (MARS) \citep{friedman1991multivariate}, smoothing spline analysis of variance (SS-ANOVA) models \citep{gu2013smoothing,wahba1990spline,wahba1995smoothing}, and component selection and smoothing operator (COSSO) \citep{lin2006component}. 
Much of the work has been restricted to additive models with only main effect functions, and potentially two-way interaction functions.
In practice, this restriction may lead to biased and inaccurate regression models. 
On the other hand, the proposed model provides a mechanism to seek relevant higher-order interaction functions by considering strong effect heredity, which rules out many impractical models from the search.}

From a statistical learning perspective, the order of functional ANOVA and resolution of representation can likely be gleaned from the collected data.  This idea is adopted in the next section to enable the construction of MRFA emulators.

\section{Estimation and Regularization}\label{fitting}

A straight-forward approach to finding a set of sets of indices $\mathcal{E}$ which obeys strong effect heredity, in both functional ANOVA and resolution, and allows construction of an accurate model is {stepwise variable selection}.
Initial investigations along these lines indicate that stepwise variable selection is capable of producing a high-accuracy model, but introduces a {very} serious computational bottleneck to model fitting, particularly for large-scale and many-input problems.
Alternatively, posing the problem as a penalized regression can provide huge computational savings.

\cite{yuan2006model} proposed the group lasso penalty to build accurate models and perform variable selection with grouped variables, for example a set of basis function evaluations.
In the group lasso framework, the overall penalty term is the sum of {unsquared} $L_2$ norms of the coefficients of variables within groups. This type of penalty ensures that all the components of the groups have zero or non-zero coefficients simultaneously.
\cite{jacob2009group} noticed that the group lasso penalty could be used to enforce a spectrum of effect hierarchies by employing an \emph{overlapping group structure}.
In particular, if a group of variables' \emph{parents} (those variables which must be present if the group is present) are always included in the unsquared $L_2$ penalty component with the group of interest, then the group of variables can only have non-zero coefficients if the parents have non-zero coefficients. 
One can consider the penalized loss function
\begin{gather}\label{eq:lossfuncion}
\begin{split}
Q=\frac{1}{n}\sum_{i=1}^n&\left(y_i-\sum_{|u|=1}^{D_{\rm max}}\sum_{r=1}^{R_{\rm max}}\sum_{k=1}^{n_u(r)}\beta_u^{rk}\varphi_u^{rk}(x_{iu})\right)^2\\
&\quad+\lambda\sum_{|u|=1}^{D_{\rm max}}\sum_{r=1}^{R_{\rm max}}\sqrt{N_u(r)\sum_{v\subseteq u}\sum_{s\le r}\sum_{k=1}^{n_v(s)}(\beta_v^{sk})^2},
\end{split}
\end{gather}
where $D_{\rm max}$ and $R_{\rm max}$ respectively denote maximal orders of functional ANOVA and resolution level, and $N_u(r)=\sum_{v\subseteq u}\sum_{s\le r}n_v(s)$.
Notably, $D_{\rm max}\ll d$ and $R_{\rm max}\ll n$ to ensure computational feasibility in a large-scale, many-input setting.
Efficient, large-scale algorithms are available for coefficient estimation in the group lasso setting \citep{meier2008group,roth2008group}.
In particular, the algorithm described in \cite{meier2008group} is implemented in the \texttt{R} \citep{R2015} package \texttt{grplasso} \citep{packagegrplasso}.

Although the algorithm in \cite{meier2008group} is quite computationally efficient, storage requirements still have potential to cause computational infeasibility, particularly for a large-scale and many-input problem. 
We propose a modification of the algorithm where \emph{candidate} basis function evaluations are added sequentially along the lasso path, as necessary to ensure effects heredity,
rather than storing all the basis functions in advance. 
The modified algorithm is given in Appendix \ref{appnd:algorithm}.
The algorithm starts from a candidate set consisting only of main effect functions with resolution level one
and
an initial penalty $\lambda_{\max}$ set as suggested in \cite{meier2008group}.
Then, the penalty parameter is gradually decreased and the model is re-fit over steps.
If the active set changes in a particular step, the candidate set is enlarged to include \emph{child} basis function evaluations as required by effect heredity in functional ANOVA and resolution.
A small value of the penalty parameter increment $\Delta$ is required to ensure that
at most one new active group is included
in each update. 
The algorithm stops when some convergence criterion is met, or alternatively memory limits are approached.

The accuracy of the emulator can depend strongly on the tuning parameter $\lambda$. When overfitting is not a major concern, for example when constructing an emulator or {near interpolator} for a deterministic computer experiment, the smallest $\lambda$ (corresponding to the most complex model) with no evidence of numeric instability 
could be taken, which in turn would give near interpolation of outputs at input locations in the data used for fitting. 
On the other hand, if overfitting is a concern, a few sensible choices for tuning parameter selection include cross-validation or classical information criteria such as Akaike information criterion (AIC) and Bayesian information criterion (BIC). 
Under some conditions, BIC is consistent for the true model when the set of candidate models contains the true model,
while AIC will select a sequence of models which are
asymptotically equivalent to the
model whose average squared error is 
smallest 
among the candidate models. Generalized cross-validation (GCV) \citep{craven1978smoothing}, leave-one-out cross-validation and AIC have similar asymptotic behavior. 
Delete-$d$ cross-validation \citep{shao1997asymptotic} is asymptotically equivalent to the generalized information criterion (GIC) with parameter $\lambda_n=n/(n-d)+1$.
See \cite{shibata1984approximate}, \cite{li1987asymptotic} and \cite{shao1997asymptotic} for more details. 
The use of AIC and BIC for regularization parameter selection in penalized regression models has been discussed in recent literature (see \cite{wang2007tuning} and \cite{zhang2010regularization}). 
\cite{wang2007tuning} showed that BIC can consistently identify the true model for the smoothly clipped absolute deviation penalty \citep{fan2001variable}, whereas the models selected by AIC and GCV tend to overfit. For the group lasso framework, our numerical results indicate AIC has slightly better performance than BIC. On the other hand, if parallel computing environments are available, cross-validation can be computationally efficient and could be used for selecting the tuning parameter $\lambda$. 

{In addition to prediction, uncertainty quantification is essential in practice. 
In Sections \ref{inference} and \ref{CIs}, we develop new theoretical results for 
consistency and inference.
Further, an algorithm for constructing pointwise confidence intervals as a means to quantifying one's statistical uncertainty in the predicted values is provided in Appendix \ref{appnd:algorithmInf}.}


\section{Consistency of the MRFA Emulator}\label{inference}
In this section, we develop new consistency results for our estimator.
Notably, these results are general and relate to the MRFA emulator only in the sense that the MRFA model forms an application case of particular interest.
The results apply to the, possibly overlapping, group lasso problem in a large $n$, large $p$ setting, and are developed along the lines described in \cite{meinshausen2009lasso} and \cite{liu2009estimation}. 
Here, we make three major contributions.
First, we extend large $n$, large $p$ lasso consistency results to the overlapping group lasso problem. 
Second, we extend the results to the case where the true function is deterministic, as is the case for many computer experiments \citep{santner2013design}. 
Third, 
we show that the results hold for situations where the responses have random noise, in addition to the deterministic response situation.

Suppose for a particular input location ${x}$, the true value of the computer model is $y({x})$.    If we 
are modeling the responses as a linear combination of basis functions $\{\varphi(\cdot)\}=\{\varphi_{u}^{rk}(\cdot):k=1,\ldots,n_u(r), r=1,\ldots,R_{\rm max},|u|=1,\ldots,D_{\rm max}\}$, but do not
make additional assumptions about $y(x)$, then we may define the best model (in an $L_2(\Omega)$ sense) as 
\begin{equation}\label{betastardef}
\beta^* = \operatorname*{argmin}_\beta \underbrace{\int_{\Omega} (y(x) - \varphi (x)^T \beta)^2 \mathrm{d} x }_{\text{oracle risk}} .
\end{equation}
This represents the oracle's choice in coefficients, knowing the exact underlying model and the entire sequence of information.  
Note that $\{\varphi(\cdot)\}$ refers to the set of basis functions, while $\varphi(x)$ refers to the vector of basis function evaluations at $x$.
The vectors of  basis function evaluations $\varphi (x)$ and corresponding coefficients $\beta$ are of length $p$, which is assumed to grow as $n$ increases, though the dependence is notationally suppressed for clarity.  This represents the natural behavior of including more basis functions in larger computer experiments.  
We assume the coefficient vector is sparse in the sense that only relatively few coefficients will be useful in predicting the underlying function.

Throughout, we consider statistical modeling in the context where 
$x_1,\ldots,x_n,\ldots$ are a sequence of input locations 
whose corresponding sequence of empirical cumulative distribution functions converges to a uniform distribution. 
In this setting, the responses can be expressed in terms of the linear model as
\begin{equation} \label{eq:lin_model_det}
y_i= \varphi(x_i)^T \beta^* + B_i, 
\end{equation}
where $B_i$ is the resulting random bias term at $x_i$. 
In the context of the MRFA model, $\varphi (x_i)$ denotes the vector of {unique} basis function evaluations at $x_i$ (i.e. not duplicate basis function evaluations appearing in the overlapping group penalty), $\beta^* \in\mathbb{R}^p$ denotes the best possible basis function coefficients, and $B_i$ denotes the left-over. 
Since the responses are not corrupted by noise, we call  this the \textit{deterministic case}.

The \textit{stochastic case} is when the computer model does not produce the same output for repeated runs at a given input.  
Stochastic computer experiments commonly use random number generators to produce difficult to predict and control internal inputs, such as customer arrival times or weather.
In the stochastic case,the responses can be expressed as
\begin{equation} \label{eq:lin_model_stoch}
y_i= \varphi(x_i)^T \beta^* + B_i +\epsilon_i,
\end{equation}
where  $\epsilon_i$ represents the random noise on the $i$th observation. 
We assume that the $\epsilon_i$s are independent, identically distributed, sub-Gaussian random variables (see Definition \ref{defnSubGaussian}) with $\mathbb{E}(\epsilon_i)=0$ and $\mathbb{V}(\epsilon_i)=\sigma^2 > 0$ for $i=1,...,n$.


{Inference is considered in the $n\rightarrow\infty,p\rightarrow\infty,p\gg n$ setting for $n$ pairs $(\varphi(x_i),y_i)_{i=1}^n$, in which the large sample distribution of the inputs $x_i$'s converges to the uniform distribution. }
The following definitions are used. For two positive sequences $a_n$ and $b_n$, we write $a_n\asymp b_n$ if, for some $C,C'>0$, $C\leqslant a_n/b_n \leqslant C'$. Similarly, we write $a_n\lesssim b_n$ if $a_n\leqslant Cb_n$ for some constant $C>0$. 
We now present the following $l_2$ consistency result, whose proof follows the logic in \cite{meinshausen2009lasso}, which is valid in both the deterministic and stochastic situation.
\begin{theorem}\label{thm:L2consistency_det}
Suppose the estimated coefficients of the overlapping group lasso are $\hat{{\beta}}$ (see (\ref{groupLassoEstimator1})) with parameter $\lambda_n$, and the best coefficients are ${\beta}^{*}$, as defined in equation (\ref{betastardef}). Let $\varphi$ be the matrix with rows $\varphi(x_i)^T$, $i=1,...,n$, {and assume the large sample distribution of the inputs $x_i$ converges to the uniform distribution.} Under assumptions on the $m$-sparse eigenvalues (Definition \ref{msparseDef} and Assumption \ref{mSparse}) of matrix $\frac{1}{n}{\varphi}^T{\varphi}$, $\lambda_n\asymp \sqrt{\frac{\log p}{n}}$, $\bar d^2=o(\log n)$, and $\|y(\cdot) - \varphi (\cdot)^T \beta^*\|_\infty = O_p(\lambda_n)$, with probability tending to 1 for $n\rightarrow \infty$,
\begin{align}\label{thmconsisteq}
\|\hat{{\beta}}-{\beta}^{*}\|^2_2\lesssim \frac{ \bar c^2 s\bar d \log p}{n},
\end{align}
where $\bar c$, $\bar d$, and $s$ denote the largest number of groups that an element of $\varphi(x_i)$ appears in, the size of the largest group, and the number of non-zero elements in unique representation ${\beta}^{*}$, respectively. 
\end{theorem}
\begin{remark}
{Note that the dimension $p$ here is allowed to increase with $n$, and consequently the number of basis functions $n_u(r)$ is also allowed to increase with $n$ since $p=\sum^{D_{\rm max}}_{|u|=1}\sum^{R_{\rm max}}_{r=1}n_u(r)$, allowing for an improving quality of approximation of $f_u$ as the sample size increases.} 
Potential dependency of $\varphi(\cdot)$, $\bar c$, $\bar d$, $s$, and $p$ on $n$ is suppressed for notational simplicity. 
{Additionally, the error variance $\sigma^2$ also influences the convergence in \eqref{thmconsisteq} but is not presented because it is treated as a constant.}  
\end{remark}





Theorem \ref{thm:L2consistency_det} demonstrates pointwise convergence of the coefficient estimates under some conditions. 
Essentially, consistent coefficient estimates are achieved if the dimension of the MRFA representation does not grow so quickly that $\log p$ is large compared to $n$. 
The $l_2$ consistency in Theorem \ref{thm:L2consistency_det} is specifically provided by two major conditions.  
The first is that the numerator of the right hand side does not grow too fast, $o(n)$.  
This in turn requires the size of groups, number of nonzero (best) coefficients, and number of groups that a variable appears in are relatively small compared with the sample size $n$.  Secondly, the bias of the model at the $i$th input $B_i$, needs to shrink quickly.  


The following corollary is an immediate consequence of Theorem \ref{thm:L2consistency_det}, and states that the oracle risk at the estimated coefficients $\hat \beta$ can be bounded in terms of the oracle risk at the best coefficients.
\begin{corollary}\label{Coro421}
Suppose the assumptions of Theorem \ref{thm:L2consistency_det} hold. The oracle risk at $\hat \beta$ can be bounded as
\begin{align}
\int_{\Omega} (y(x) - \varphi (x)^T \hat{\beta})^2 \mathrm{d} x \lesssim  \frac{ \bar c^2 s\bar d \log p}{n}.
\end{align}
\end{corollary}

\begin{remark}
{A related upper bound on the oracle risk is derived by \cite{juditsky2000functional}, in which the functional aggregation problem is considered, where the best combination of basis functions with coefficients in a convex compact subset of the $l_1$-ball is considered as the optimality target.
In our problem, we consider a larger class of functions when defining optimality,
which allows us to obtain a faster convergence rate.} 
\end{remark}









\section{Confidence intervals}\label{CIs}
This section develops and discusses theory  for the large sample distribution of a decorrelated score statistic \citep{ning2014general} that can be used to form confidence intervals for the stochastic case in (\ref{eq:lin_model_stoch}).  
A modification of this technique leveraging Apley's coverage correction \citep{apley2017spuq} is proposed for the deterministic case (\ref{eq:lin_model_det}), and has good coverage and interval width in our numeric examples.
Confidence intervals in the stochastic case are considerably easier.  The authors are not able to confirm similar results for the deterministic case. The end of this section will explain a modification that yielded good behavior in the deterministic examples we studied.



A pointwise confidence interval under the stochastic case (\ref{eq:lin_model_stoch}) is constructed by inverting a one-dimensional hypothesis test of $H_0: y^*({x}) =\delta$, as provided in Theorem \ref{thm:MainThmLinear}, after the model has been reparametrized so that $y^*({x})$ equals a particular coefficient in the model.
The one-dimensional hypothesis test uses a decorrelated score function, that converges weakly to standard normal, following \cite{ning2014general}.
Details are provided below and in Appendix \ref{ProofThmmainLinear}.



Without loss of generality, suppose the parameter of interest is $\beta_1 \in\mathbb{R} $,  and the remaining coefficients are nuisance parameters $\beta_{-1} = (\beta_2,\ldots,\beta_p)^T\in\mathbb{R}^{p-1}$. 
Then the linear model \eqref{eq:lin_model_stoch} can be written as $y_i=\beta_1 \varphi_{i1} +\beta_{-1}^T\varphi_{i,-1}+B_i + \epsilon_i$, where $\varphi_{i,-1}=(\varphi_{i2},\ldots,\varphi_{ip})^T$. Following \cite{ning2014general}, define a \textit{decorrelated} score function
\begin{align*}
S(\beta_{1},\beta_{-1})=-\frac{1}{n\sigma^2}\sum_{i=1}^n(y_i-\beta_1 \varphi_{i1} - \beta_{-1}^T\varphi_{i,-1})(\varphi_{i1}-{w}^T\varphi_{i,-1}),
\end{align*}
where ${w}=\mathbb{E} (\varphi_{i,-1}\varphi_{i,-1}^T)^{-1}\mathbb{E} (\varphi_{i,-1}\varphi_{i1})$. The score function for the target parameter has been decorrelated with the nuisance parameter score function. Here, the full parameter vector ${\beta}$, consisting of target and nuisance parameters $\beta_{1}$ and ${\beta_{-1}}$, can be estimated via the original overlapping group lasso problem, so that $\hat\beta=(\hat \beta_{1},\hat \beta_{-1}^T)^T$.
On the other hand, ${w}$ can be estimated via
\begin{equation}\label{eq:estimatew}
\hat{{w}} =\arg\min\|{w}\|_1, \mbox{ s.t. } \bigg\|\frac{1}{n}\sum_{i=1}^n \varphi_{i,-1}(\varphi_{i1}-{w}^T\varphi_{i,-1})\bigg\|_2\leqslant \lambda'
\end{equation}
and the error variance $\sigma^2$ can be estimated by a consistent estimator $\hat{\sigma}^2$.
Note that $\lambda'$ is another tuning parameter. 
The minimization is on the $l_1$ norm of $w$, since we want to ensure sparsity of $\hat w$.
Let $\beta_{1}^*$ and $\beta_{-1}^*$ denote the values of $\beta_{1}$ and $\beta_{-1}$ which minimize the oracle risk defined in (\ref{betastardef}). 
The following (one-dimensional) inference result can be obtained. A proof is provided in Appendix \ref{ProofThmmainLinear}.
\begin{theorem}\label{thm:MainThmLinear}
Under $H_0:\beta_{1}^*=\beta_{1,0}$, $\lambda'\asymp\sqrt{\frac{\log p}{n}}$, $\sigma^2>0$, and the assumptions of Theorem \ref{MainThmLinear},
\begin{align*}
\sqrt{n}\hat{S}_{\hat{\sigma}^2}(\beta_{1,0},\hat\beta_{-1})\hat{I}_{\beta_{1}|\beta_{-1}}^{-1/2}\stackrel{\rm dist.}{\longrightarrow} \mathcal{N}(0,1),
\end{align*}
where $\hat{I}_{\beta_{1}|{\beta_{-1}}}=\frac{1}{n\hat{\sigma}^2}\sum^n_{i=1}\varphi_{i1}(\varphi_{i1}-\hat{{w}}^T\varphi_{i,-1})$, and $\hat{S}_{\hat{\sigma}^2}(\beta_{1},\beta_{-1})=-\frac{1}{n\hat{\sigma}^2}\sum_{i=1}^n(y_i-\beta_{1} \varphi_{i1}-\beta_{-1}^T\varphi_{i,-1})(\varphi_{i1}-\hat{w}^T\varphi_{i,-1}).$
\end{theorem}



The solution to optimization problem \eqref{eq:estimatew} can 
also be represented as
\begin{equation}\label{eq:estimatew2}
\hat{{w}} =\arg\min\bigg\|\frac{1}{n}\sum_{i=1}^n \varphi_{i,-1}(\varphi_{i1}-{w}^T\varphi_{i,-1})\bigg\|^2_2 + \lambda''\|{w}\|_1,
\end{equation}
where $\lambda''$ is a transformed tuning parameter. 
Notice that this is 
a lasso problem where the $j$-th response, $j=1,\ldots,p-1$, is $(\frac{1}{n}\sum_{i=1}^n \varphi_{i,-1}\varphi_{i1})_j$ and the covariance matrix is $\frac{1}{n}\sum_{i=1}^n \varphi_{i,-1}\varphi_{i,-1}^T$. The tuning parameter $\lambda''$ can be selected via cross-validation, aiming for a minimal sum of squared errors, or simply fixed.
Theorem \ref{thm:MainThmLinear} requires all the assumptions of Theorem \ref{thm:L2consistency_det}. 
In addition, 
it is required that the smallest eigenvalue of $\mathbb{E}(\varphi_{i,-1}\varphi_{i,-1}^T)$ 
is bounded away from zero, 
the number of nonzero elements in ${w}=\mathbb{E} (\varphi_{i,-1}\varphi_{i,-1}^T)^{-1}\mathbb{E} (\varphi_{i,-1}\varphi_{i1})$ is  small compared to $n$, and
the tail probabilities of residuals and basis function evaluations are small in the sense that they are sub-Gaussian. 

Note that Theorem \ref{thm:MainThmLinear} is not directly applicable to the deterministic case, since Theorem \ref{thm:MainThmLinear} requires that the error variance is non-zero. 
In the deterministic case, the $\sqrt{\log p/n}$ bias decay dominates the large sample behavior.
For more detail, see Appendix \ref{ProofThmmainLinear}. 
We propose using a small constant instead of $\hat \sigma^2$ for the deterministic case, which provides a conservative confidence interval. 



Now, we re-express the linear model \eqref{eq:lin_model_stoch} to obtain pointwise confidence interval on predictions.
Let ${\varphi}^*$ denote the basis function evaluations at a particular predictive location $x^*$, and $y^*$ denote the predictive output, $y^*={\beta}^T{\varphi}^*$. 
By extending ${\varphi}^*$ to a basis of $\mathbb{R}^p$, $A = ({\varphi}^*,c_2,\ldots,c_{p})$, the linear model \eqref{eq:lin_model_stoch} can be written as $y_i=\eta_1\tilde\varphi_{i1}+{\eta}_{-1}^T \tilde\varphi_{i,-1}+B_i + \epsilon_i$, where $(\tilde\varphi_{i1},\tilde\varphi_{i,-1}^T)^T=A^{-1}{\varphi}_i$ and $(\eta_1,{\eta}_{-1}^T)^T = A^T\beta$. Thus, the hypothesis test $H_0:y^*=\eta_{10}$ is equivalent to $H_0:\eta_1=\eta_{10}$, and a $(1-\alpha)\times 100\%$ confidence interval on $y^*$ can be constructed by inverting the hypothesis test, as stated in the following corollary. 
An algorithm for confidence interval construction is provided in Appendix \ref{appnd:algorithmInf}.
In the algorithm, a simple construction for the matrix $A$ is to take $c_i$ as a unit vector with $i$th element equaling one. 
{Note that after the transformation with this choice of $A$, the assumptions of Theorem \ref{thm:MainThmLinear} still hold.} 
Then, the inverse of $A$ can be computed efficiently  via partitioned matrix inverse results \citep{harville1997matrix}.

\begin{corollary}\label{cor:CI}
Under the assumptions of Theorem \ref{thm:MainThmLinear},
a $(1-\alpha)\times 100\%$ confidence interval on $y^*$ can be constructed as 
\[
\left\{y^*| \Phi^{-1}\left(\frac{\alpha}{2}\right)\leqslant \sqrt{n}\hat S_{\hat{\sigma}^2}(y^*,\hat{\eta}_{(-1)})\hat {I}_{y^*|{\eta}_{(-1)}}^{-1/2} \leqslant \Phi^{-1}\left(1-\frac{\alpha}{2}\right)\right\},
\]
where $\hat{I}_{y^*|{\eta}_{(-1)}}=\frac{1}{n\hat{\sigma}^2}\sum^n_{i=1}\tilde\varphi_{i1}(\tilde\varphi_{i1}-\hat{{w}}^T\tilde\varphi_{i,-1})$, $\hat S_{\hat{\sigma}^2}(y^*,\eta_{(-1)})=\frac{1}{n\hat{\sigma}^2}\sum_{i=1}^n(y_i-y^* \tilde\varphi_{i1}-\eta_{(-1)}^T\tilde\varphi_{i,-1})(\tilde\varphi_{i1} - \hat{w}^T\tilde\varphi_{i,-1})$, $\Phi$ is the cumulative distribution function of the standard normal distribution, and $\hat{\eta}_{-1}$ is an estimator of $\eta_{-1}$, which can be obtained by plugging in the estimator of $\beta$.
\end{corollary}

Optimization problems \eqref{eq:estimatew} and equivalently \eqref{eq:estimatew2} can be very computationally challenging when $n$ is large. 
In particular, for $R_{\rm max}=10$ and $D_{\rm max}=10$ (as used in the examples later), $p$ is nearly $10^7$, making storage of the $\varphi_{i,-1}:(p-1)\times 1$, $i=1,\ldots,n$ infeasible without specialized computational resources.
In Appendix \ref{appnd:algorithmInfalternative}, we provide a large $n$ modification to the confidence interval algorithm in Appendix \ref{appnd:algorithmInf}.
In the modification, only those nuisance basis function evaluations which have been included for consideration up to the selected stage of the group lasso problem are considered in $\varphi_{i,-1}$, reducing the size of $\varphi_{i,-1}$ by several orders of magnitude.
Given the reduced $\varphi_{i,-1}$, 
we propose to estimate $w$ via a ridge regression, since
sparsity of $w$ relative to the sample size $n$ is ensured by default for this reduced dimensional nuisance parameter set.
While the intervals are computationally feasible in a large scale, many-input setting, their coverage is somewhat liberal.
For the deterministic case (\ref{eq:lin_model_det}),
we can apply a post-hoc correction, as proposed by \cite{apley2017spuq}.
The idea is to regard $\sigma^2$ as a tuning parameter and then apply a cross-validation method to the confidence intervals constructed by Corollary \ref{cor:CI} to find the $\sigma^2$ which most closely achieves the nominal coverage $(1-\alpha)\times 100\%$. 

An illustration of these pointwise confidence intervals is shown in Figure \ref{fig:CI}. 
In the example, the true function is $f(x)=\exp(-1.4x)\cos(3.5\pi x)$, shown as a black dotted line, and we attempt to build an emulator using 14 evenly spaced data points {between 0 and 1}, shown as black dots. Consider a very simple MRFA model, with {{three levels of resolution and}} Wendland's kernel candidate basis functions with $k=2$, {shown as light gray in Figure \ref{fig:CI}}. 
{The left panel considers a stochastic case}, where the output values are sampled from $y=f(x)+\epsilon$, and the $\epsilon$ are independent, identically normally distributed with mean zero and standard deviation $\sigma=0.3$. 
The MRFA emulator for (penalized regression) tuning parameter $\lambda=0.647$, which is chosen via cross-validation, is shown as the solid blue line, and the 95\% confidence intervals are shown as the gray shaded region in Figure \ref{fig:CI}, with the consistent estimate of $\sigma^2$, $\hat{\sigma}^2=\frac{1}{n-s}\sum^n_{i=1}(y_i-\hat{\beta}^T{{\varphi}}(x_i))^2$, and the tuning parameter $\lambda''$ chosen via cross-validation at each untried input site of interest. 
{Although the MRFA emulator deviates from the true mean function, the confidence intervals are able to quantify the deviation and contain the true mean values.}
Given a set of testing samples of size 500, 95.4\% of the true mean function values are contained by the confidence interval, which achieves close to the nominal coverage 95\%. 
The right panel considers a deterministic case (i.e., without the noise $\epsilon$). 
The MRFA emulator for a small tuning parameter $\lambda=0.001$ is shown as the solid blue line and the 95\% confidence intervals are shown as the gray shaded region, with the post-hoc correction for the estimate of $\sigma^2$ proposed by \cite{apley2017spuq}.
With the post-hoc correction, the deterministic case 
confidence intervals achieve good performance using the same techniques developed in this section. 
The MRFA emulator almost interpolates every data point, and, importantly, the confidence intervals are able to quantify the model bias and contain the true values. 
Given a set of testing samples of size 500, 95.8\% of true values are contained by the confidence intervals, which achieves close to the nominal coverage 95\%.

\begin{figure}[h]
\centering
\includegraphics[width=\textwidth]{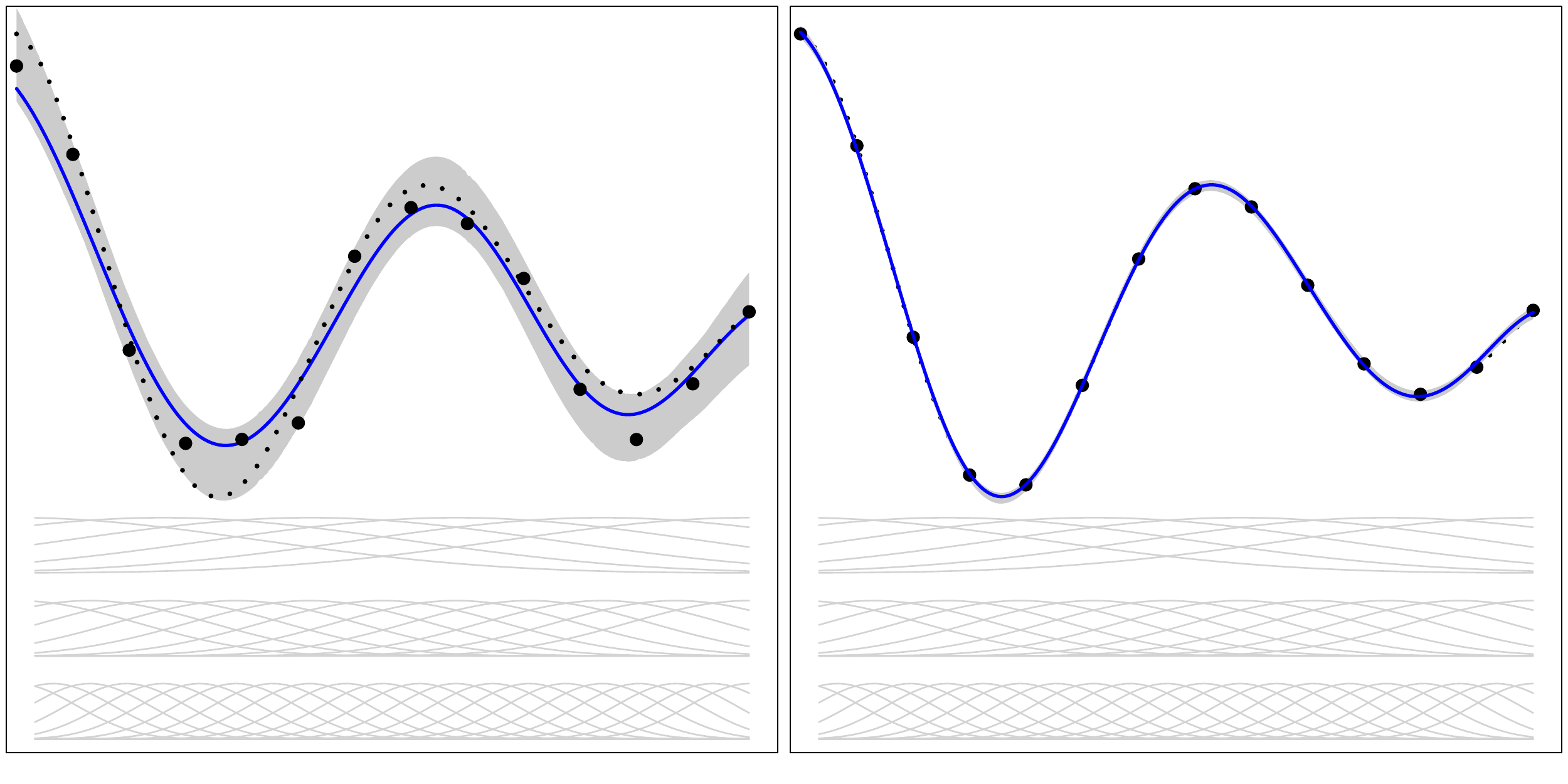}
\caption{Illustration of confidence intervals for stochastic (left) and deterministic (right) cases. Black dotted line represents the true function, black dots represent the collected data, and the MRFA emulator is represented as the blue lines, whose candidate basis functions are shown in solid light gray, with the gray shaded region providing a pointwise 95\% confidence band.}
\label{fig:CI}
\end{figure}



\section{Basis function selection}\label{basisFunctions}


Basis functions of a given input dimension should be selected so that they are capable of approximating a broad spectrum of practically encountered target functions, with flexibility increasing as the level of resolution increases.

For a particular dimensionality of component function $m=|u|$,
a reasonable building block for a set of basis functions is a positive definite function. The function $\phi:\mathbb{R}^{m}\to\mathbb{R}$ is positive definite if $\sum_{i,j}\alpha_i\alpha_j\phi(x_i-x_j)\ge 0$ for any $\alpha_i\in\mathbb{R}$, $x_i\in\mathbb{R}^{m}$ and strictly positive for distinct $x_i$ if at least one $\alpha_i$ is non-zero.
These could be constructed by integrating the full-dimensional kernel over margins as indicated in Theorem \ref{thm:basis}.
More simply, the kernels could be selected to ensure a desired smoothness of the target component functions.
Common example kernels include the Mat$\acute{\rm e}$rn and squared exponential correlation functions. We use Wendland's kernels \citep{wendland1995piecewise} in the examples presented here.
Notably, Wendland's kernels are compactly supported, potentially enabling construction of a sparse design matrix, which can in turn provide computational and numeric advantages. 
Wendland's kernels are based on evaluating inter-point distances in positive polynomials truncated to $[0,1]$ and otherwise zero.
The parameter $k$ determines the smoothness at zero ($2k$ continuous derivatives). 
The polynomial terms of Wendland's kernels are computed recursively based on the parameter $k$ and the dimension of input $m$. 

The \emph{center} and \emph{scale} of these basis functions, or \emph{kernels}, can be adjusted via $c$ and $h$, respectively, in the representation $\phi((x-c)/h)$.
For a particular resolution level, a straightforward choice is to take as basis functions a set of kernels with centers well-spread through the input space. 
The scale should be chosen large enough to ensure the desired smoothness of the target function, but not so large that numeric issues arise in parameter estimation. 
The number of centers, and in turn coefficients, concretely describes the complexity of the resolution level.
Take as an example the 5 basis functions shown in light gray in the left panel of Figure \ref{fig:BasisRepresentation}. 
With centers $0,0.25,\ldots,1$ and width $0.75$,
these 5 basis functions are capable of approximating a broad range of relatively smooth and slowly varying target functions. 
For the next resolution level, the same basic kernel can be used again, but with a denser set of centers and correspondingly smaller scale.
Take once again the example basis functions shown in Figure \ref{fig:BasisRepresentation}. 
The 10 second-level 
resolution basis functions with centers $0,0.11,\ldots,1$ and width $0.5$ augment the first-level resolution basis functions to allow approximation of an even broader range of target functions.
Note that for a fixed dimensionality $m$ and resolution level $r$ the span of these basis functions forms a linear subspace of the RKHS associated with kernel $\phi((\cdot-\cdot)/h_r)$, where $h_r$ denotes the bandwidth for the {highest} (or finest) resolution level $r$.
Another reasonable choice for basis functions could be polynomials of increasing degree.

\section{Examples}\label{examples}

Several examples are examined in this section, a ten-dimensional, large-scale example which demonstrates the 
algorithm and statistical inference, a larger-scale and many-input example with a relatively complicated underlying function, and a stochastic function example.
A few popular test functions are examined additionally. 
These examples show that the multi-resolution functional ANOVA typically substantially outperforms traditional Gaussian process methods in terms of computational time, emulator accuracy, model interpretability, and scalability. 
In addition, we also compare with the local Gaussian process method, which is a scalable method proposed by \cite{gramacy2015local}.
All the numerical results were obtained using \texttt{R} \citep{R2015} on a server with 2.3 GHz CPU and 256GB of RAM.
The traditional Gaussian process, local Gaussian process and MRFA approaches were compared and respectively implemented in \texttt{R} packages \texttt{mlegp} \citep{mlegp}, \texttt{laGP} \citep{gramacy2015lagp} and \texttt{MRFA} \citep{sung2017mrfa}. 
The default settings of the packages \texttt{mlegp} and \texttt{MRFA} were selected.
For the package \texttt{laGP}, initial values and maximum values for correlation  parameters were given as suggested in \cite{gramacy2015lagp}.
For \texttt{laGP} and \texttt{MRFA}, 10 CPUs were requested via \texttt{foreach} \citep{pkgforeach} for parallel computing.

In the implementation of the MRFA model, Wendland's kernels with $k=2$ are chosen, and at most 10-way interaction effects and 10 resolution levels are considered ($R_{\max}=10$ and $D_{\max}=10$). For the tuning parameter $\lambda$, in 
Sections \ref{sec:example1}, \ref{sec:example2} and \ref{sec:morefunctions} where the target functions are deterministic, the smallest $\lambda$, corresponding to the most complex model, without exceeding memory allocation is taken. In Section \ref{sec:example3} where a stochastic target is considered, AIC, BIC and CV criteria were considered for choosing the tuning parameter and the comparison is explicitly discussed.

\subsection{10-dimensional data set}\label{sec:example1}
Consider a 10-dimensional, uniformly distributed input set of size $n$ in a $[0,1]^{10}$ design space and $n_{\rm test}=10,000$ random predictive locations generated from the same design space. 
The deterministic target function
\[
f(x_1,\ldots,x_{10})=\sin(1.5 x_1\pi)+3\cos(3.5 x_2\pi)+5\exp(x_3)+2\cos(x_2\pi)\sin(x_3\pi)
\]
is considered. 
Note that changes in $x_3$ have a relatively large influence on the output. 
Further, $x_1,x_2$ and $x_3$ are active while $x_4,\ldots,x_{10}$ are inert. Table \ref{Tab:procedure_simple} presents the selected inputs by MRFA in the fitted model for $n=1,000$. The main effect of $x_3$ with resolution level one is first entertained, and in the final fitted model ($\lambda=0.003$) the influential inputs are correctly selected while the irrelevant inputs ($x_4,\ldots,x_{10}$) are also identified (in the sense that they do not appear in the fitted model). Noticeably, our algorithm finds the basis functions which obey strong effect heredity in the final fitted model. In particular, $\hat f_{\{3\},1},\hat f_{\{2\},1}$, and $\hat f_{\{2,3\},1}$ are selected in the final fitted model.

\begin{table}[h]
\centering
\begin{tabular}{c|c}
\hline  $\lambda$ &  Selected inputs\\
\hline
1904.819 & $\hat f_{\{3\},1}$\\ 
1885.866 & $\hat f_{\{3\},1},\hat f_{\{2\},1}$\\ 
551.225 & $\hat f_{\{3\},1},\hat f_{\{2\},1},\hat f_{\{1\},1}$\\ 
87.544 & $\hat f_{\{3\},1},\hat f_{\{2\},1},\hat f_{\{1\},1},\hat f_{\{2,3\},1}$\\ 
\vdots & \vdots\\ 
0.003 & $\hat f_{\{3\},1},\hat f_{\{2\},1},\hat f_{\{1\},1},\hat f_{\{2,3\},1},\hat f_{\{2\},2},\hat f_{\{3\},2},\hat f_{\{1\},2},\hat f_{\{2,3\},2},\hat f_{\{2\},3},\hat f_{\{3\},3},\hat f_{\{1\},3},\hat f_{\{2,3\},3}$\\ 
\hline
\end{tabular}
\caption{Selected effects and resolution by model complexity.}
\label{Tab:procedure_simple}
\end{table}

Table \ref{Tab:simplefunction} shows the performance of MRFA based on designs of increasing size $n$, in comparison to \texttt{mlegp} and \texttt{laGP}. The fitting time of \texttt{laGP} is not shown in the example (and the ones in the following sections) because the fitting process of the approach cannot be simply separated from prediction. Note that \texttt{mlegp} is only feasible at $n=1,000$ in the numerical study, so results for $n>1000$ are not reported.  In contrast, it can be seen that MRFA is feasible and accurate for large problems. Furthermore, it is \textit{much} faster to fit and predict from and, even in cases when traditional Gaussian process fitting is feasible, more accurate. In this example with several inert input variables, compared to local Gaussian process fitting, even though \texttt{laGP} is feasible for large problems, the accuracy of the emulators is not comparable with traditional Gaussian process fitting or MRFA. In particular, MRFA can improve the accuracy at least 10000-fold over the considered sample sizes and it is even faster than local Gaussian process fitting in the cases $n=1,000$ and $n=10,000$. 
In addition, in all examples, the true active variables (i.e., $x_1,x_2,x_3$ and the interaction effect) are correctly selected, while all inactive variables (i.e., $x_4,\ldots,x_{10}$) are excluded. 
This example demonstrates that the MRFA method is capable of not only providing an accurate emulator at a much smaller computational cost, but also identifying important variables, which can be useful for model interpretation. 

\begin{table}[h]
\centering
\begin{tabular}{|C{2cm}|C{2cm}|C{2cm}C{2cm}C{2cm}C{2cm}|}
\hline  & \multirow{2}{*}{$n$} & Fitting  & Prediction  & RMSE & Variable\\
 &  & time (sec.) & time (sec.) & $(\times 10^{-5})$ & detection\\
\hline \texttt{mlegp} & 1,000 & 1993 & 158 & 40.81 & -\\
\hline \multirow{4}{*}{\texttt{laGP}} &  1,000 & - & 318 & 172998 & -\\
      & 10,000 & - & 331 & 71027& -\\
     & 100,000 & - & 331 & 20437& -\\
   & 1,000,000 & - & 361 & 6893 & -\\
\hline \multirow{4}{*}{\texttt{MRFA}} &  1,000 & 44  &  6 & 3.24 & 100\%\\
      & 10,000 & 124 & 5 & 1.14  & 100\%\\
     & 100,000 & 1325 & 5 & 0.72& 100\%\\
   & 1,000,000 & 61515 & 74 & 0.38& 100\%\\
\hline
\end{tabular}
\caption{Performance of 10-dimensional example with $n_{\rm test}=10,000$ random predictive locations.}
\label{Tab:simplefunction}
\end{table}

To demonstrate the statistical inference results and techniques discussed in Section 5, confidence intervals on emulator predictions are compared. 
The evaluation includes coverage rate, average width of intervals, and average interval score \citep{gneiting2007strictly}. Coverage rate is the proportion of the time that the interval contains the true value, while interval score combines the coverage rate and the width of intervals, 
\[
S_{\alpha}(l,u;x)=(u-l)+\frac{2}{\alpha}(l-x)\mathbbm{1}\{x<l\}+\frac{2}{\alpha}(x-u)\mathbbm{1}\{x>u\},
\]
where $l$ and $u$ are the lower and upper confidence limits, and $(1-\alpha)\times 100\%$ is the confidence level. Note that a smaller score corresponds to a better interval.

Continuing the above example, we consider 95\% confidence intervals. 
Here, we consider the large $n$ modification to the confidence interval algorithm with the reduced dimensional nuisance parameter, as given in Appendix \ref{appnd:algorithmInfalternative}. 
The unmodified algorithm performs similarly for $n=1,000$ and $n=10,000$, but is not feasible for the larger sample sizes.
The results of the evaluations are given in Table \ref{Tab:CIresults}.  It can be seen that the MRFA intervals have
coverage rate close to the nominal coverage 95\%, while \texttt{mlegp} yields very poor intervals that are both wide and contain less than 80\% of the true values. While \texttt{laGP} has reasonable coverage, it yields very wide confidence intervals, which result in a poor interval score. In contrast, the confidence intervals of MRFA perform best in terms of the interval score, given their small width.
Notably, the technique of \cite{apley2017spuq} could also be applied to \texttt{mlegp} and \texttt{laGP} to bring their coverage near target, but their widths would still be much larger than MRFA.


\begin{table}[h]
\centering
\begin{tabular}{|C{2cm}|C{2cm}|C{2.5cm}C{3cm}C{3cm}|}
\hline & \multirow{2}{*}{$n$} & Coverage & Average width& Average interval\\
 & & rate (\%) &  ($\times 10^{-5}$)&  score ($\times 10^{-5}$)\\
\hline \texttt{mlegp} & 1,000 & 75.09 & 6139.29 & 6313.56\\
\hline \multirow{4}{*}{\texttt{laGP}} &  1,000 & 82.18 & 313469 & 1324927\\
      & 10,000 & 92.85 & 126172 & 392917\\
     & 100,000 & 93.54 & 53313 & 85058\\
   & 1,000,000 & 93.20 & 24073 & 31478\\
\hline \multirow{4}{*}{\texttt{MRFA}} &  1,000 & 100.00 &  27.39 & 27.39\\
      & 10,000 & 98.56 & 3.12 & 3.39\\
     & 100,000 & 95.84 & 2.31 & 3.06 \\
   & 1,000,000 & 97.69 & 1.64 & 1.94 \\
\hline
\end{tabular}
\caption{Performance of prediction intervals in the 10-dimensional example with $n_{\rm test}=10,000$ random predictive locations.}
\label{Tab:CIresults}
\end{table}

\subsection{Borehole function}\label{sec:example2}
In this subsection, we use a relatively complex target function for a variety of input dimensions to further examine the MRFA in a many-input context. The borehole function \citep{zacks1998modern} represents a model of water flow through a borehole, and has input-output relation
\[
f(\mathbf{x})=\frac{2\pi T_u(H_u-H_l)}{\ln (r/r_w)(1+\frac{2LT_u}{\ln (r/r_w)r^2_wK_w}+\frac{T_u}{T_l})},
\]
where $r_w\in[0.05, 0.15]$ is the radius of borehole (m), $r\in[100, 50000]$ is the radius of influence (m), $T_u\in[63070, 115600]$ is the transmissivity of upper aquifer (m$^2$/yr), $H_u\in[990, 1110]$ is the potentiometric head of upper aquifer (m), $T_l\in[63.1, 116]$ is the transmissivity of lower aquifer (m$^2$/yr), $H_l\in[700, 820]$ is the potentiometric head of lower aquifer (m), $L\in[1120, 1680]$ is the length of borehole (m), and $K_w\in[9855, 12045]$ is the hydraulic conductivity of borehole (m/yr). 
Here, all inputs are rescaled to the unit hypercube.

Similar to the setup in the previous subsection, $n$ training locations along with $n_{\rm test}=10,000$ predictive locations are randomly generated from a uniform distribution on $[0,1]^d$.
Notice in the borehole experiment, there are eight active variables.
We include $d-8$ irrelevant variables for demonstration. 
Table \ref{Tab:circuit} shows the performance of traditional Gaussian process, local Gaussian process, as well as MRFA based on designs of increasing size $n$ and input dimension $d$. 
For a fixed $d$, the MRFA is feasible and accurate for large problems, while traditional Gaussian process fitting is only feasible for the experiment of size $1,000$. Note that the accuracy for $n=1,000,000$ can be further improved if more memory allocation is in hand. 
Alternatives for the case where model fitting exceeds a user's limited budget are discussed in Section \ref{discussion}. In addition, in cases when traditional Gaussian process fitting is feasible, the fitting and prediction procedure of MRFA is \textit{much} faster while retaining the accuracy (in some cases MRFA is much more accurate, see $d=20$ and $60$). Similar to the results in the previous subsection, local Gaussian process fitting is feasible for large problems, but it is less accurate than both traditional Gaussian process and MRFA.
With increasing $d$, the performance of MRFA varies only slightly, while traditional Gaussian process and local Gaussian process fitting perform \textit{substantially} worse with larger $d$ in terms of time cost and accuracy. This result is not surprising, since the irrelevant inputs are screened out (or equivalently, the influential inputs are identified) by our proposed algorithm, as demonstrated in Section \ref{sec:example1}. 
Notice that the $d=20$ \texttt{mlegp} example has very poor accuracy. 
This example was explored quite extensively and for several random number seeds.
In all cases, the 
likelihood function was
highly ill-conditioned,
resulting in very low accuracy.
This numerical issue was also pointed out in \cite{GPfit}.

\begin{table}[!h]
\centering
\begin{tabular}{
|C{1cm}|C{2cm}|C{2cm}|C{2cm}C{2cm}C{2cm}|}
\hline
\multirow{2}{*}{$d$} & \multirow{2}{*}{Method} & \multirow{2}{*}{$n$}& Fitting  & Prediction & RMSE \\
& && Time (sec.) & Time (sec.) & \\
\hline \multirow{9}{*}{10} & \texttt{mlegp} & 1,000 & 9405 &  99 & 0.5406\\
\hhline{~-----} & \multirow{4}{*}{\texttt{laGP}} & 1,000 & - & 324 & 2.2541  \\
   & &  10,000 & - & 327 & 1.0952\\
   & & 100,000 & - & 326 & 0.5316\\
   &  & 1,000,000 & - & 343 & 0.2667\\
\hhline{~-----} & \multirow{4}{*}{\texttt{MRFA}} & 1,000 & 344 & 31 & 0.5659  \\
   & &  10,000 & 858 & 15 & 0.1777\\
   & & 100,000 & 8753 & 72 & 0.1186\\
   &  & 1,000,000 & 160326 & 179 &  0.0901*\\
\hline \multirow{9}{*}{20} & \texttt{mlegp} & 1,000 & 12358  & 172 & 16.4539\\
\hhline{~-----} & \multirow{4}{*}{\texttt{laGP}} & 1,000 & - & 356 &  10.1838 \\
   & &  10,000 & - & 359 & 9.7302\\
   & & 100,000 & - & 362 & 10.0245\\
   &  & 1,000,000 & - & 429 & 9.3887\\
\hhline{~-----} & \multirow{4}{*}{\texttt{MRFA}} & 1,000 & 278 & 24 & 0.5583 \\
   & &  10,000 & 786 & 14 &0.1853\\
   & & 100,000 & 8443 & 67 & 0.1220\\
   & & 1,000,000 & 254457 & 214 & 0.0924*\\
\hline \multirow{9}{*}{60} & \texttt{mlegp} & 1,000 & 15999 & 186  & 3.5841\\
\hhline{~-----} & \multirow{4}{*}{\texttt{laGP}} & 1,000 & - & 599 & 20.6825  \\
   & &  10,000 & - & 600 & 34.3782\\
   & & 100,000 & - & 638 & 45.3728\\
   &  & 1,000,000 & - & 924 & 51.2694\\
\hhline{~-----} & \multirow{4}{*}{\texttt{MRFA}} & 1,000 & 534 & 26 & 0.7034 \\
   & &  10,000 & 812 & 15 & 0.1770\\
   & & 100,000 & 6482 & 50  & 0.1312\\
   & & 1,000,000 & 150477 & 90 & 0.0980*\\
\hline
\end{tabular}
\caption{The borehole example with $n_{\rm test}=10,000$ random predictive locations. *Note that due to memory limits, in these cases $R_{\text{max}}=3$ and $D_{\text{max}}=3$ are considered instead.}
\label{Tab:borehole}
\end{table}

\subsection{Stochastic Function}\label{sec:example3}

In this subsection, a stochastic function is considered. In particular, this example demonstrates tuning parameter selection. We consider the following function, which was used in \cite{gramacy2009adaptive},
\begin{equation}\label{eq:gramacylee}
f(x_1,x_2,x_3,x_4,x_5,x_6)=\exp\left\{\sin([0.9\times (x_1+0.48)]^{10})\right\}+x_2x_3+x_4+\epsilon,
\end{equation}
where $\epsilon\sim\mathcal{N}(0,0.05^2)$ and $x_i\in[0,1],i=1,\ldots,6$. The function is nonlinear in $x_1,x_2$ and $x_3$, and linear in $x_4$. In $x_1$, it oscillates more quickly as it reaches the upper bound of the interval $[0,1]$. $x_5$ and $x_6$ are irrelevant variables. 

Here, we consider 5 replicates at each unique training location, $n=5m$, as indicated in \cite{wangStochKrigingPaper}, 
along with $n_{\rm test}=10,000$ unique predictive locations randomly generated from a uniform distribution on $[0,1]^d$. Since the choice of tuning parameter $\lambda$ in \eqref{eq:lossfuncion} can be particularly crucial in stochastic function emulation, we consider AIC, BIC and 10-fold CV as selection criteria. For the implementation of 10-fold CV, 10 CPUs are requested for parallel computing.
Table \ref{Tab:circuit} shows the performance of traditional Gaussian process, local Gaussian process, and MRFA with these three selection criterion based on designs of increasing size $n$. It can be seen that, similar to the results in the previous subsections, traditional Gaussian process is only feasible at $n=1,000$, while MRFA is feasible and accurate for large problems. Even when traditional Gaussian process is feasible, MRFA is much faster in terms of fitting and prediction, and more accurate with any tuning parameter selection method. Local Gaussian process fitting is feasible for large problems, but less accurate than MRFA and traditional Gaussian process.
Among the three criteria, it can be seen that AIC, BIC and CV have relatively small differences {in terms of prediction accuracy}. 
Computationally, the tuning parameters can be chosen within 2 seconds using AIC or BIC, while the computational costs of CV can be considerable.

This example also illustrates the flexibility of the proposed method. 
From \eqref{eq:gramacylee}, the function appears not to satisfy the strong effect heredity conditions, because the main effects of $x_2$ and $x_3$ are not present. 
On the other hand, the function can be easily re-expressed in a form that does satisfy strong effect heredity. 
For example,  
\[
f(x_1,\ldots,x_6)=-1+\exp\left\{\sin([0.9\times (x_1+0.48)]^{10})\right\}+x_2+x_3+(x_2-1)(x_3-1)+x_4+\epsilon,
\]
which satisfies the strong effect heredity assumption because main effect functions of $x_2$ and $x_3$ appear in the function in addition to the interaction function $(x_2-1)(x_3-1)$.

\begin{table}[h]
\centering
\begin{tabular}{
|C{2cm}|C{2cm}|C{2cm}C{2cm}C{1cm}C{2cm}C{2cm}|}
\hline
 & \multirow{2}{*}{$n$}& Fitting & Prediction&  & Selection & RMSE \\
 && Time (sec.)& Time (sec.) &  & Time (sec.) & $(\times 10^{-1})$\\
\hline  \texttt{mlegp} & 1,000 & 2524 & 88 &&  & 1.64\\
\hline \multirow{4}{*}{\texttt{laGP}} & 1,000 & - &394&  &  & 7.30 \\
 & 10,000 & - &439 &   &   & 6.07\\
 & 100,000& - &457&   &  & 4.70\\
  & 1,000,000& - & 433&   &  & 3.85\\
\hline
 \multirow{12}{*}{\texttt{MRFA}} & \multirow{3}{*}{1,000} & \multirow{3}{*}{96}& \multirow{3}{*}{8}& AIC & 1 & 1.36 \\
&   & & & BIC & 1 & 1.36\\
 & & & & CV & 92 & 1.32\\
\hhline{~------}
& \multirow{3}{*}{10,000} & \multirow{3}{*}{443} &\multirow{3}{*}{23} & AIC & 1&  0.18 \\
&   & & & BIC & 1& 0.19\\
 & & & & CV & 423 & 0.26 \\
\hhline{~------}
 & \multirow{3}{*}{100,000} & \multirow{3}{*}{2999} & \multirow{3}{*}{34}& AIC &1 & 0.14  \\
&   & & & BIC & 1& 0.14\\
 & & & & CV & 2213 & 0.14\\
\hhline{~------}
& \multirow{3}{*}{1,000,000} & \multirow{3}{*}{61504}& \multirow{3}{*}{103}& AIC &  1& 0.01  \\
&   & & & BIC & 1& 0.01\\
 & & & & CV & 55849 &0.05\\
\hline
\end{tabular}
\caption{The 6-dimensional stochastic function example with $n_{\rm test}=10,000$ random predictive locations.}
\label{Tab:circuit}
\end{table}

\subsection{Other Functions}\label{sec:morefunctions}
In this subsection, we present three more example functions in comparison with \texttt{laGP} and \texttt{mlegp}, the 3-dimensional bending function \citep{plumlee2017lifted}, the 6-dimensional OTL circuit function \citep{ben2007modeling}, and the 10-dimensional wing weight function \citep{forrester2008engineering}. 
The details of these examples and their input ranges are given in Appendix \ref{append:morefunctions}.

The comparison results are shown in Table \ref{Tab:otherexample}. Similar to the results in the previous subsections, the results indicate the MRFA outperforms the traditional Gaussian process in terms of prediction accuracy, except for the wing function at $n=1,000$ where the traditional Gaussian process fitting has better accuracy. The reason might be that the underlying wing weight function contains high-order interaction functions making it not particularly well-suited to low-order representation. See \eqref{eq:wingfunction} in the Appendix. Nevertheless, even when the traditional Gaussian process fitting is feasible (at $n=1,000$), the MRFA is much faster than traditional Gaussian process fitting.  Local Gaussian process fitting is feasible for large problems and has better accuracy in the low-dimensional example (see Table \ref{Tab:otherexample}(a)), but it is less accurate in the other two examples and in some cases slower than the MRFA.

\begin{table}[!h]
\begin{subtable}{\linewidth}\centering
{\begin{tabular}{|C{2cm}|C{2cm}|C{2cm}C{2cm}C{2cm}|}
\hline \multirow{2}{*}{$d=3$} & \multirow{2}{*}{$n$} & Fitting  & Prediction  & RMSE\\
 &  & time (sec.) & time (sec.) & $(\times 10^{-5})$ \\
\hline \texttt{mlegp} & 1,000 & 1807 & 140 & 5.64\\
\hline \multirow{4}{*}{\texttt{laGP}} &  1,000 & - & 310 & 0.66 \\
      & 10,000 & - & 312 & 0.21\\
     & 100,000 & - & 311 & 0.08\\
   & 1,000,000 & - & 316 &  0.04\\
\hline \multirow{4}{*}{\texttt{MRFA}} &  1,000 & 49  &  8 & 2.16 \\
      & 10,000 & 293 & 14 & 0.46  \\
     & 100,000 & 3311 & 25 & 0.20 \\
   & 1,000,000 & 113279 & 159 &  0.14*\\
\hline
\end{tabular}}
\caption{Performance of the 3-dimensional bending function. *Note that due to memory limits, in the cases $R_{\text{max}}=3$ and $D_{\text{max}}=3$ are considered instead.}
\end{subtable}

\begin{subtable}{\linewidth}\centering
{\begin{tabular}{|C{2cm}|C{2cm}|C{2cm}C{2cm}C{2cm}|}
\hline \multirow{2}{*}{$d=6$} & \multirow{2}{*}{$n$} & Fitting  & Prediction  & RMSE\\
 &  & time (sec.) & time (sec.) & $(\times 10^{-4})$ \\
\hline \texttt{mlegp} & 1,000 & 3976 & 173 & 13.70\\
\hline \multirow{4}{*}{\texttt{laGP}} &  1,000 & - & 314 & 102.71\\
      & 10,000 & - & 301 & 27.01\\
     & 100,000 & - & 323 & 11.43\\
   & 1,000,000 & - & 328 & 4.80\\
\hline \multirow{4}{*}{\texttt{MRFA}} &  1,000 & 294  &  19 & 7.81 \\
      & 10,000 & 798 & 17 & 2.05  \\
     & 100,000 & 6688 & 82 & 1.42 \\
   & 1,000,000 & 122075 & 133 & 1.18* \\
\hline
\end{tabular}}
\caption{Performance of the 6-dimensional OTL circuit function. *Note that due to memory limits, in the cases $R_{\text{max}}=3$ and $D_{\text{max}}=3$ are considered instead.}
\end{subtable}

\begin{subtable}{\linewidth}\centering
{\begin{tabular}{|C{2cm}|C{2cm}|C{2cm}C{2cm}C{2cm}|}
\hline \multirow{2}{*}{$d=10$} & \multirow{2}{*}{$n$} & Fitting  & Prediction  & RMSE\\
 &  & time (sec.) & time (sec.) & $(\times 10^{-1})$ \\
\hline \texttt{mlegp} & 1,000 & 2922 & 228 & 1.56\\
\hline \multirow{4}{*}{\texttt{laGP}} &  1,000 & - & 327 & 19.74\\
      & 10,000 & - & 325 & 10.72\\
     & 100,000 & - & 329 & 5.04\\
   & 1,000,000 & - & 347 & 2.22\\
\hline \multirow{4}{*}{\texttt{MRFA}} &  1,000 & 1319  &  28 & 7.77 \\
      & 10,000 & 1633 & 21 & 1.52  \\
     & 100,000 & 12289 & 84 & 1.39 \\
   & 1,000,000 & 168854 & 148 &  1.18*\\
\hline
\end{tabular}}
\caption{Performance of the 10-dimensional wing weight function. *Note that due to memory limits, in the cases $R_{\text{max}}=1$ and $D_{\text{max}}=3$ are considered instead.}
\end{subtable}

\caption{Performance of the bending, OTL circuit, and wing weight functions with $n_{\rm test}=10,000$ random predictive locations.}
\label{Tab:otherexample}
\end{table}

\section{Discussion}\label{discussion}
While large-scale and many-input nonlinear regression problems have become typical in the modern ``big data'' context, Gaussian process models are often impractical due to memory and numeric issues. In this paper, we proposed a multi-resolution functional ANOVA (MRFA) model, which targets a low resolution representation of a low order functional ANOVA, with respect to strong effect heredity, to form an accurate emulator in a large-scale and many-input setting. Implementing a forward-stepwise variable selection technique via the group lasso algorithm, the representation can be efficiently identified without supercomputing resources. Moreover, we provide new theoretical results regarding consistency and inference for a potentially overlapping group lasso problem, which can be applied to the MRFA model. Our numerical studies demonstrate that our proposed model not only successfully identifies influential inputs, but also provides accurate predictions for large-scale and many-input problems with a much faster computational time compared to traditional Gaussian process models.

The MRFA model has a similar flavor to multivariate adaptive regression splines (MARS) \citep{friedman1991multivariate}.
On the other hand, the flexibility in basis function choice along resolution levels, forward-stepwise variable selection via group lasso, and confidence interval development for the MRFA are quite different. 
Moreover, empirical studies in \cite{ben2007modeling} show the Gaussian process outperforming MARS in terms of prediction accuracy, while our numerical studies show MRFA outperforming Gaussian process.

The proposed MRFA indicates several avenues for future research.
First, when the sample size is too large due to a user's limited budget (e.g., memory limitation), sub-sampling methods can be naturally applied to the MRFA approach. For example, \cite{breiman1999pasting} proposed \textit{pasting Rvotes} and \textit{pasting Ivotes} methods, which use random sampling and importance sampling, respectively. Moreover, \textit{m-out-of-n bagging} (also known as \textit{subagging}) \citep{buchlmann2002analyzing,buja2006observations,friedman2007bagging} uses sub-samples for aggregation and might be expected to have similar accuracy to bagging, which uses bootstrap samples to improve the accuracy of prediction \citep{breiman1996bagging}. These sub-sampling methods provide the potential
to extend the MRFA model to even larger data sets.

Next, if the basis functions are constructed by integrating the full-dimensional kernel over margins as indicated in Theorem \ref{thm:basis}, 
one may consider the native space norm with kernel $\Phi$ instead of the 2-norm in the penalized loss function \eqref{eq:lossfuncion}. 
In fact, both norms were examined in our numeric studies and the results indicated that the penalized loss function with respect to the native space norm may increase computational costs without much improvement in prediction accuracy. For example, for the 10-dimensional example in Section 6.1, with $n=1,000$, the fitting with the native space norm costs about 6 minutes while fitting with the 2-norm only costs 44 seconds, and both result in roughly the same RMSE.

Last but not least, it is conceivable that the MRFA approach can be generalized to a non-continuous, for example binary, response. 
One might proceed by replacing the residual sum of squares in \eqref{eq:lossfuncion} by the corresponding negative log-likelihood function, and extending the group lasso algorithm to other exponential families, as done in \cite{meier2008group}. The inference results, however, cannot be directly applied to a non-continuous response. 

\section*{Appendices}
\renewcommand{\thesubsection}{\Alph{subsection}}
\renewcommand{\theequation}{\Alph{subsection}.\arabic{equation}}

\subsection{Proof of Theorem \ref{thm:basis}}\label{append:basis}
First, a useful lemma is given.

\begin{lemmmm}\label{lemma:fANOVA}
Denote $\mathcal{F}_u = \{\int_{\Omega_{-u}} \left(f(x) - \sum_{v \subset u} f_v (x) \right) {\rm d}x_{-u}|f\in \mathcal{N}_\Phi,f_v\in\mathcal{F}_v\}$. Suppose $\Phi\in\Omega\times\Omega\rightarrow\mathbb{R}$ is a symmetric positive-definite kernel on $\Omega=[0,1]^d$ and $\Phi$ is a product kernel. Then, $$f_u\in\mathcal{F}_u = \{f_v + g_u |g_u \in \mathcal{N}_{\Phi_{u}}, v \subset u, f_v \in \mathcal{F}_v \},$$
where $\Phi_{u}=\prod_{j \in u} \phi_j$.
\end{lemmmm}
\begin{proof}
Initially consider a finite element. 
The proof proceeds by induction. For $u=\emptyset$, we have that if $f\in\mathcal{N}_{\Phi}$, then 

\[
f_{\emptyset} =  \int_{\Omega} f(x) {\rm d}x=\int_{\Omega}\sum_{y \in X} \beta_y \Phi(x, y) {\rm d}x=\sum_{y \in X} \beta_y\int_{\Omega} \Phi(x, y) {\rm d}x:=\alpha\in\mathbb{R}.
\]
This shows $f_\emptyset\in \mathcal{F}_\emptyset = \{f(\cdot) = \alpha| \alpha \in \mathbb{R} \} $.


Let $f_u\in \mathcal{F}_u$ for any $|u|\leq k$. Note that $\int_{\Omega_{-u}}  {\rm d}x_{-u}=1$ for any $u$, since $\Omega=[0,1]^d$. 
Thus, for $|u'|=k+1$, 
\begin{align*}
 f_{u'}(x) &= \int_{\Omega_{-u'}} \left(f(x) - \sum_{v \subset u'} f_v (x) \right) {\rm d}x_{-u'} =\int_{\Omega_{-u'}} f(x) {\rm d}x_{-u'} - \sum_{v \subset u'} f_v (x)\\
  & =  \sum_{y \in X}\beta_y  \int_{\Omega_{-u'}}   \Phi (x, y) {\rm d}x_{-u'}- \sum_{v \subset u'} f_v (x) \\
  & = \sum_{y \in X}\beta_y  \int_{\Omega_{-u'}}   \prod^d_{j=1}\phi (x_j, y_j) {\rm d}x_{-u'} - \sum_{v \subset u'} f_v (x)\\
 &= \sum_{y \in X}  \beta_y  \prod_{j \in u'}\phi_j (x_j, y_j)  \int_{\Omega_{-u'}} \prod_{j \notin u'}  \phi_{j} (x_j, y_j){\rm d}x_{-u'}-\sum_{v \subset u'} f_v (x)\\
 &= \sum_{y \in X}  \tilde{\beta}_y \prod_{j \in u'} \phi_i (x_i,y_i)   - \sum_{v \subset u'} f_v (x),
\end{align*}
where $\tilde{\beta}_y=\beta_y\int_{\Omega_{-u'}} \prod_{j \notin u'}  \phi_{j} (x_j, y_j){\rm d}x_{-u'}.$ Hence, since 
$\sum_{y \in X}  \tilde{\beta}_y  \phi_{u'} (\cdot,y_i) \in \mathcal{N}_{\Phi_{u'}}$ and $f_v\in \mathcal{F}_v$ for any $|v|\leq k$, we have $f_{u'}\in \mathcal{F}_{u'} = \{f = f_v + g_{u'} |g_{u'} \in \mathcal{N}_{\Phi_{u'}}, v \subset u', f_v \in \mathcal{F}_v \}$. Therefore, by induction, $f_u\in \mathcal{F}_u = \{f_v + g_u |g_u \in \mathcal{N}_{\Phi_u}, v \subset u, f_v \in \mathcal{F}_v \}$ is true for any $u\subseteq D$.

Since any element of an RKHS is bounded \citep{aronszajn1950theory}, we may use the dominated convergence theorem \citep{bartle2014elements} to interchange the integral and the limit of the finite sums to extend to an arbitrary element.
\end{proof}

By Lemma \ref{lemma:fANOVA}, we have $f(x)=\sum_{u\subseteq D}f_u(x)$, where $f_u(x)\in\mathcal{F}_u=\{f_v + g_u |g_u \in \mathcal{N}_{\Phi_{u}}, v \subset u, f_v \in \mathcal{F}_v \}$. Thus, by the fact that $g^{(1)}_u+g^{(2)}_u\in\mathcal{N}_{\Phi_{u}}$ for $g^{(1)}_u,g^{(2)}_u\in\mathcal{N}_{\Phi_{u}}$, $f(x)$ can be represented as
$f(x)=\sum_{u\subseteq D}f_u(x)$, where $f_u \in \mathcal{N}_{\Phi_{u}}$.

\subsection{Algorithm for Estimation}\label{appnd:algorithm}
\begin{enumerate}
\item[1.] Let $\mathcal{A}$ denote the set of active groups and $\mathcal{C}$ the set of candidate groups. Start with $\mathcal{A}=\emptyset$ and $\mathcal{C}=\{(u,r)|u=\{1\},\ldots,\{d\},r=1 \}$.
Set an initial penalty $\lambda_{\max}$ and a small increment $\Delta$.	  
\item[2.] Set up an overlapping group lasso algorithm which minimizes the penalized likelihood function
\begin{displaymath}
\frac{1}{n}\sum_{i=1}^n\left(y_i-\sum_{(u,r)\in\mathcal{C}}\sum_{k=1}^{n_u(r)}\beta_u^{rk}\varphi_u^{rk}(x_{iu})\right)^2+\lambda\sum_{(u,r)\in\mathcal{C}}\sqrt{N_u(r)\sum_{v\subseteq u}\sum_{s\le r}\sum_{k=1}^{n_v(s)}(\beta_v^{sk})^2}.
\end{displaymath}
Denote the input-output function as $\hat{{\beta}}_{\lambda}=\texttt{grplasso}(\lambda,\mathcal{C},\hat{{\beta}}_{\lambda+\Delta})$. The inputs include a penalty value $\lambda$, the candidate set $\mathcal{C}$ and the estimated coefficient with penalty value $\lambda+\Delta$, and the output $\hat{{\beta}}_{\lambda}$ is the corresponding estimated coefficient by the algorithm. 
Start with $\lambda=\lambda_{\max}$ and $\hat{{\beta}}_{\lambda+\Delta}={0}$.

\item[3.] Do $\hat{{\beta}}_{\lambda}=\texttt{grplasso}(\lambda,\mathcal{C},\hat{{\beta}}_{\lambda+\Delta})$ and obtain the set of active groups $\mathcal{A}'\subseteq\mathcal{C}$ based on $\hat{{\beta}}_{\lambda}$. 
Set $\lambda=\lambda-\Delta$.
If $\mathcal{A}'\setminus\mathcal{A}\ne\emptyset$, then $\mathcal{A}\leftarrow\mathcal{A}'$ and $\mathcal{C}\leftarrow\mathcal{C}\cup\mathcal{C}',$ where $\mathcal{C}'$ contains the new candidate groups necessary to satisfy strong effects heredity given the updated $\mathcal{A}$.
	       
\item[4.] Repeat step 3 until some convergence criterion is met.
\end{enumerate}

\subsection{Confidence Interval Algorithm}\label{appnd:algorithmInf}
\begin{enumerate}
\item Let ${\varphi}^*$ denote the basis function evaluations at a particular predictive location $x^*$. Extend ${\varphi}^*$ to a basis of $\mathbb{R}^p$ and denote it as $A=({\varphi}^*,c_2,\ldots,c_{p})$. Compute $(\tilde{Z}_i,\tilde{{Q}}_i)^T=A^{-1}{\varphi}_i$ for $i=1,\ldots,n$ and $(\hat{\eta}_1,\hat{\eta}^T_{(-1)})= A^T\hat{\beta}_{\lambda}$, where $\hat{\beta}_{\lambda}$ is the estimated coefficient with penalty $\lambda$.
\item Compute the estimated decorrelated score function
\begin{align*}
\hat S(0,\hat{\eta}_{(-1)})=-\frac{1}{n\hat\sigma^2}\sum_{i=1}^n(y_i- \hat{\eta}_{(-1)}^T\tilde{{Q}}_i)(\tilde Z_i-\hat{w}^T\tilde{Q}_i),
\end{align*}
where
\begin{equation*}
\hat{{w}} =\arg\min\bigg\|\frac{1}{n}\sum_{i=1}^n {\tilde Q}_i(\tilde Z_i-{w}^T{\tilde Q}_i)\bigg\|_2 + \lambda''\|{w}\|_1,
\end{equation*}
and $\hat{\sigma}^2$ is a consistent estimator of $\sigma^2$. For example, $\sigma^2$ can be estimated by $\hat{\sigma}^2=\frac{1}{n-s}\sum^n_{i=1}(y_i-\hat{\beta}^T_{\lambda}{\varphi}_i)^2$, where $s$ the the number of non-zero elements in $\hat{\beta}_{\lambda}$. Another estimator is the cross-validation based variance estimator. Define the $K$ cross-validation folds as $\{D_1,\ldots,D_K\}$ and compute
\[
\hat{\sigma}^2=\min_{\lambda}\frac{1}{n}\sum^K_{k=1}\sum_{i\in D_k}(y_i-(\hat{\beta}_{\lambda}^{(-k)})^T{\varphi}_i)^2,
\]
where $\hat{\beta}_{\lambda}^{(-k)}$ is the overlapping group lasso estimate at $\lambda$ over the data after the $k^{th}$ fold is omitted. This estimator has been used for the variance estimation in lasso regression problems. See \cite{fan2012variance}.
\item Compute the interval
\begin{align*}
[c_{\alpha/2}/b,c_{1-\alpha/2}/b], 
\end{align*}
where $c_{\alpha/2}=-\hat S(0,\hat{\eta}_{(-1)}) +\sqrt{\frac{b}{n}}\Phi^{-1}(\alpha/2)$, $c_{1-\alpha/2}=-\hat S(0,\hat{\eta}_{(-1)}) +\sqrt{\frac{b}{n}}\Phi^{-1}(1-\alpha/2)$, $b=\frac{1}{n\hat \sigma^{2}}\sum_{i=1}^{n}\tilde Z_i(\tilde Z_i-\hat w^T\tilde Q_i)$. By some algebraic manipulation, one can show that this interval is same as the one in Corollary \ref{cor:CI}.
\end{enumerate}

\subsection{Confidence Interval Algorithm Modification for Large $n$}\label{appnd:algorithmInfalternative}
\begin{enumerate}

\item In Algorithm \ref{appnd:algorithmInf}, replace $\tilde{Q}_i$ by $\tilde{Q}_{*i}$ and $p$ by $p_*$, where the nuisance $\varphi_{ij}$, $j=1,\ldots,p_*$ only contain basis functions in the candidate groups at the selected $\lambda$, say $\mathcal{C}_\lambda$.

\item Replace $\hat{w}$ by
\begin{equation}\label{eq:alternativewhat}
\hat{w}_*=\left(\sum^n_{i=1}\tilde{Q}_{*i}\tilde{Q}^T_{*i}+\eta I_{p_*-1}\right)^{-1}\left(\sum^n_{i=1}\tilde{Q}_{*i}\tilde{Z}_i\right)
\end{equation}
with a small positive $\eta$, where $I_{p_*-1}$ is a $(p_*-1)\times(p_*-1)$ identity matrix.

\item For the deterministic case (\ref{eq:lin_model_det}),
\begin{enumerate}
\item[(i)] Define $K$ cross-validation folds as $\{D_1,\ldots,D_K\}$ and partition the original samples $\{x_i,y_i\}^n_{i=1}$ via the $k$ folds.

\item[(ii)] Regard $\hat{\sigma}^2$ in Algorithm \ref{appnd:algorithmInf} as an unknown parameter. Let $\hat{u}^{(-k)}(x^*,\hat{\sigma}^2)$ and $\hat{l}^{(-k)}(x^*,\hat{\sigma}^2)$ be the upper and lower limits at a predictive location $x^*$ by Algorithm \ref{appnd:algorithmInf} over the data after the $k^{th}$ fold is omitted, respectively. 

\item[(iii)] Replace $\hat{\sigma}^2$ by
\[
\hat{\sigma}^2_*=\arg\min_{\hat{\sigma}^2}\left|\left(\frac{1}{n}\sum^K_{k=1}\sum_{i\in D_k} \mathbbm{1}\{y_i \in [\hat{l}^{(-k)}(x_i,\hat{\sigma}^2),\hat{u}^{(-k)}(x_i,\hat{\sigma}^2)]\}\right)-(1-\alpha)\right|,
\]
where $\mathbbm{1}\{A\}$ is an indicator function of the set $A$.
\end{enumerate}
\end{enumerate}

\subsection{Proof of Theorem \ref{thm:L2consistency_det}}\label{proofThmL2}

\subsubsection{Notation and Reformulation}\label{AppReform}

First, we introduce some additional notation. 
For a matrix ${M}=[M_{jk}]$, let $\|{M}\|_{\max}=\max_{j,k}|M_{jk}|$, $\|{M}\|_1=\sum_{j,k}|M_{jk}|$, and $\|{M}\|_{l_\infty}=\max_j \sum_k |M_{jk}|$. 
For ${v}=(v_1,...,v_p)^T\in \mathbb{R}^p$, and $1\leqslant q< \infty$, 
define $\|{v}\|_q=(\sum_{i=1}^p|v_i|^q)^{1/q}$. Define $\|{v}\|_0=|\{i:v_i\neq 0\}|$. For $S\subseteq\{1,...,p\}$, let ${v}_S=\{v_j:j\in S\}$ and $\bar S$ be the complement of $S$. Given $a,d\in\mathbb{R}$, we use $a\vee b$ and $a\wedge b$ to denote the maximum and minimum of $a$ and $b$.

For convenience, we restate the loss function as follows. 
Consider groups $J_1,...,J_{p_n}$, where $J_j\subseteq \{1,...,p\}$, 
and $\bigcup_{j=1}^{p_n}J_j=\{1,...,p\}$. Notice that we do not require $J_{j_1}\bigcap J_{j_2}=\emptyset$. 
Define $C_k=\{j:k\in J_j\}$ and $c_k=|C_k|$. 
Thus, $C_k$ is the set of indices of the groups variable $k$ belongs to and $c_k$ is the number of groups that variable $k$ belongs to. We can also treat $c_k$ as replicates of index $k$. For notational simplicity, in the proof we write $\hat \beta_n$ and $\beta_n^*$ as $\hat \beta$ and $\beta^*$, respectively. We also write $\varphi_n(X_i)$ as $\varphi_i$ for simplicity.
Define the vector of variable $k$ coefficients over all groups in which it appears
${\beta}^Z_{kC_k}=(\beta_{kj_{k1}},\ldots,\beta_{kj_{kc_k}})^T$,
where $j_{kl}$ denotes the index of variable $k$ within the $l^{\rm th}$ group in which it appears,  
and the vector of all coefficients
${\beta}^Z=(({\beta}^Z_{1C_1})^T,\ldots,({\beta}^Z_{pC_p})^T)^T$.
Let ${\beta}_{J_j}=(\beta_{kj})^T_{k\in J_j}$, where $\beta_{kj}$ is the coefficient of the $k^{th}$ variable and $k$ is in $j^{th}$ group. 
Let $d_j=|J_j|$. Consider the following optimization problem
\begin{align}\label{groupLassoEstimator}
\hat{{\beta}}^{Z,\lambda_n}=\arg \min_{{\beta}^Z}\bigg\{\frac{1}{2n}\sum_{i=1}^n(y_i-\sum_{k=1}^{p}\bigg(\sum_{m=1}^{c_k}\beta_{kj_{km}}\bigg)\varphi_{ki})^2+\lambda_n\sum_{j=1}^{p_n}\sqrt{d_j}\|{\beta}_{J_j}\|_2\bigg\},
\end{align}
where $\lambda_n$ is a positive number.
We define the overlapping group lasso estimator as
\begin{align}\label{groupLassoEstimator1}
\hat {{\beta}}^{\lambda_n} = \bigg(\sum_{k=1}^{c_1}\hat\beta_{1j_{1k}}^{\lambda_n},...,\sum_{k=1}^{c_p}\hat\beta_{pj_{pk}}^{\lambda_n}\bigg)^T,
\end{align}
in which we stress $\lambda_n$ since it will influence the solution of (\ref{groupLassoEstimator}). Notice that by this definition, the least squares term becomes $\frac{1}{2n}\sum_{i=1}^n(y_i-{\varphi}_i^T \hat {{\beta}}^{\lambda_n} )^2$, which is the same as in original group lasso case. We use $\frac{1}{2n}$ instead of $\frac{1}{n}$ for brevity of the Karush-–Kuhn–-Tucker (KKT) conditions, which are as following.\par
\begin{prop}\label{KKTprop}
Let ${\varphi}$ be the matrix with rows ${\varphi}_i^T$, $i=1,\ldots,n$. Let ${\psi}_j$ denote the $j^{th}$ column of ${\varphi}$, for $j=1,\ldots,p$. Necessary and sufficient conditions for $\hat{{\beta}}^Z$ to be a solution to (\ref{groupLassoEstimator}) are
\begin{align*}
-\frac{1}{n}{\psi}_j^T({y}-{\varphi}\hat {{\beta}}^{\lambda_n} )+\frac{\lambda_n \sqrt{d_k} \hat {{\beta}}^{\lambda_n}_{jk}}{\|\hat {{\beta}}^{\lambda_n}_{J_k}\|_2}=0,&\qquad\forall j\in J_k\;{\rm with}\;   \hat {{\beta}}^{\lambda_n}_{J_k}\neq 0\\
\|-\frac{1}{n}{\psi}_j^T({y}-{\varphi}\hat {{\beta}}^{\lambda_n})\|_2\leqslant \lambda_n\sqrt{d_k},&\qquad\forall j\in J_k\;{\rm with}\;   \hat {{\beta}}^{\lambda_n}_{J_k} = 0.
\end{align*}
\end{prop}
The following lemma \cite{liu2009estimation} states that at most $n$ groups can be nonzero.
\begin{lemmmm}\label{lemma1Han}
Suppose $\lambda_n>0$, a solution $\hat{{\beta}}^{Z,\lambda_n}$ exists such that the number of nonzero groups $|S(\hat{{\beta}}^{Z,\lambda_n})|\leqslant n$, the number of data points, where $S({\beta})=\{J_j:\hat {{\beta}}_{J_j}\neq 0\}$.
\end{lemmmm}
\begin{proof}
The proof of Lemma 1 in \cite{liu2009estimation} is also valid here.
\end{proof}
By Lemma \ref{lemma1Han}, for brevity, sometimes we say $\hat {{\beta}}^{\lambda_n}$ with $|S(\hat{{\beta}}^{Z,\lambda_n})|\leqslant n$, 
which is derived by combining (\ref{groupLassoEstimator}) and (\ref{groupLassoEstimator1}), is the solution of (\ref{groupLassoEstimator}). We will also write $\|{y}-{\varphi}{\beta}\|_2^2$ 
instead of $\sum_{i=1}^n\bigg(y_i-\sum_{k=1}^{p}\bigg(\sum_{m=1}^{c_k}\beta_{kj_{km}}\bigg)\varphi_{ki}\bigg)^2$. Let $\bar c=\max_j\{c_1,...,c_p\}$ and $\bar d=\max_j\{d_1,...,d_{p_n}\}$, the maximum number of groups a variable appears in and maximum group size, respectively. Let $s$ be the number of nonzero elements in $\beta^*$ and 
$p$ be the dimension of $\beta^*$. Notice that $s$ and $p$ (as well as $\bar c$ and $\bar d$) can depend $n$. 

\subsubsection{Proof of Theorem \ref{thm:L2consistency_det}}
Our proof follows a similar line to \cite{meinshausen2009lasso}, but extends their results to the overlapping group lasso. We only need to show the stochastic case. The deterministic case is true because the proof is still valid by taking $\epsilon= 0$. 
A sketch of the proof is as follows. 
We first define the coefficients obtained from the de-noised model as a de-noised estimator. Then, by showing the difference between the de-noised estimator and true coefficients, and the difference between de-noised estimator and the estimator obtained via overlapping group lasso are both small, we obtain $l_2$ convergence. All the proofs of the lemmas in this section are in Appendix \ref{PfofLemmasC}.\par 
Before we state and prove the main result, we introduce a definition which is useful in the proof.
\begin{defn}\label{denoiseDefn}
Denote $y(\xi)={\varphi}{\beta}^{*}  +\xi(\epsilon + \delta)$ as a de-noised model with level $\xi$ $(0\leqslant \xi \leqslant 1)$, we define
\begin{align}\label{DenoisedGroupLasso}
\hat{{\beta}}^{\lambda,\xi}=\arg\min_\beta \frac{1}{2n}\|y(\xi)-{\varphi}{\beta}\|_2^2+\lambda_n\sum_{j=1}^{p_n}\sqrt{d_j}\|{\beta}_{J_j}\|_2
\end{align}
to be the de-noised estimator at noise level $\xi$, where $\hat{{\beta}}^{\lambda,\xi}$ is defined similarly as in (\ref{groupLassoEstimator1}).
\end{defn}
In order to characterize the eigenvalues of a matrix under sparsity, we introduce the following definition, which can be found in \cite{meinshausen2009lasso}.
\begin{defn}\label{msparseDef}
The $m$-sparse minimum and maximum eigenvalue of a matrix ${C}=\frac{1}{n}{\varphi}^T{\varphi}$ are $\phi_{\min}(m)=\min_{{\beta}:\|{\beta}\|_0\leqslant m} \frac{{\beta}^T{C}{\beta}}{{\beta}^T{\beta}}$ and $\phi_{\max}(m)=\max_{{\beta}:\|{\beta}\|_0\leqslant m} \frac{{\beta}^T{C}{\beta}}{{\beta}^T{\beta}}$. Also, denote $\phi_{\max}=\phi_{\max}((s\bar c+n)\bar{d})$ where $s$, $\bar c$, and $\bar{d}_n$ are defined as in section \ref{AppReform}.
\end{defn}
Now we introduce an assumption concerning $\phi_{\min}(\cdot)$ and $\phi_{\max}$. Detailed discussion has been shown in \cite{meinshausen2009lasso}.
\begin{assum}\label{mSparse}
There exist constants $0<\kappa_{\min}\leqslant \kappa_{\max}<\infty$ such that \\ $\lim\inf_{n\rightarrow \infty} \phi_{\min}(s\bar c\bar{d}\max\{\log n,\bar c\})\geqslant \kappa_{\min}$ and $\lim\sup_{n\rightarrow \infty} \phi_{\max}\leqslant \kappa_{\max}$.
\end{assum}
For continuity, we repeat Theorem 4.1 here.

\noindent {\bf Theorem 4.1.}
Under Assumption \ref{mSparse}, if $\lambda_n\asymp \sigma\sqrt{\frac{\log p}{n}}$, $\bar d^2=o(\log n)$, and $\|y(\cdot) - \varphi (\cdot)^T \beta^*\|_\infty = O_p(\lambda_n)$, for the (overlapping) group lasso estimator constructed in (\ref{groupLassoEstimator}) and (\ref{groupLassoEstimator1}), with probability tending to 1 for $n\rightarrow \infty$,
\begin{align*}
\|\hat{{\beta}}^{\lambda_n}-{\beta}^*\|^2_2\lesssim \frac{ \bar c^2 s\bar d \log p}{n}.
\end{align*}
Let ${\beta}^{\lambda_n}=\hat{{\beta}}^{\lambda_n,0}$. The $l_2$-consistency can be obtained by bounding the bias and variance terms, i.e.
\begin{align*}
\|\hat{{\beta}}^{\lambda_n} -{\beta}^*\|_2^2\leqslant 2\|\hat{{ \beta}}^{\lambda_n} - {\beta}^{\lambda_n}\|_2^2 + 2\|{\beta}^{\lambda_n} - {\beta}^*\|_2^2.
\end{align*}

\begin{remark}
The condition $\|y(\cdot) - \varphi (\cdot)^T \beta^*\|_\infty = O_p(\lambda_n)$ implies $B_i = O_p(\lambda_n)$. In the proof of Theorem 4.1, the condition $B_i = O_p(\lambda_n)$ is sufficient.
\end{remark}

Let $T=\{t:{\beta}^*_i\neq 0,\beta^*_{it} \mbox{ is a component of }{\beta}^{Z*}\}$ represent the set of indices for all the groups with possibly nonzero coefficient vectors. 
Let $s_n=|T|$. Thus, $s_n\leqslant s\bar c$. The solution ${\beta}^{\lambda_n}$ can, for each value of ${\lambda_n}$, be written as ${\beta}^{\lambda_n}={\beta}^*+{\gamma}^{\lambda_n}$, where ${\gamma}^{\lambda_n}$ is defined as the solution of the following optimization problem:
\begin{align}\label{OptGamma}
\arg\min_{{\gamma}} & \quad f({\gamma},{\gamma}^Z)&\nonumber\\
\mbox{s.t. } & \sum_{k=1}^{c_i}\beta^Z_{ik}=\beta^*_i, \quad i=1,...,p;\\
             & \sum_{k=1}^{c_i}\gamma^{Z}_{ij_{ik}}=\gamma_i, \quad i=1,...,p,\nonumber
\end{align}
where
\begin{align*}
f({\gamma},{\gamma}^Z) & = n{\gamma}^T{A}{\gamma} + \lambda_n\sum_{t\in T^c}\sqrt{d_t}\|{\gamma}^Z_t\|_2+\lambda_n\sum_{t\in T}\sqrt{d_t}(\|{\gamma}^Z_t+{\beta}^Z_t\|_2-\|{\beta}^Z_t\|_2),
\end{align*}
where ${A}=\frac{1}{n}{\varphi}^T{\varphi}$. This optimization problem is obtained by plugging ${\beta}^*+{\gamma}^{\lambda_n}$ into (\ref{DenoisedGroupLasso}). Notice the $\arg\min$ problem is with respect to ${\gamma}$ instead of $({\gamma},{\gamma^Z})$.\par

Next, we state a lemma which bounds the $l_2$-norm of ${\gamma}^{\lambda_n}$.
Its proof is provided in Appendix \ref{proofboundGamma}.
\begin{lemmmm}\label{boundGamma}
Under Assumption \ref{mSparse}, with a positive constant $C$, the $l_2$-norm of ${\gamma}^{\lambda_n}$ is bounded for sufficiently large values of $n$ by $\|{\gamma}^{\lambda_n}\|_2\leqslant \frac{\lambda_n  \sqrt{\bar c s_n\bar d}}{n}\bigg/\bigg(\sqrt{\frac{\kappa_{\min}}{2}(1-\frac{4\bar d}{\log n})}-\sqrt{\frac{2\kappa_{\max}\bar d^2}{\log n}}\bigg)$.
\end{lemmmm}
Now, we bound the variance term. For every subset $M\subset \{1,...,p\}$ with $|M|\leqslant n$, denote $\hat{{ \theta}}^M\in \mathbb{R}^{|M|}$ the restricted least square estimator of the noise $\epsilon$,
\begin{align}\label{equation12inlemma2Han}
\hat {{\theta}}^M=({\varphi}_M^T{\varphi}_M)^{-1}{\varphi}_M^T(\epsilon + B),
\end{align}
where $B = (B_1,..,B_n)^T$ and $\epsilon = (\epsilon_1,..,\epsilon_n)^T$.
Now we state lemmas, which bound the $l_2$-norm of this estimator, and are also useful for the following parts of this development. 
First we define sub-exponential variables, sub-exponential norms, sub-Gaussian variables, and sub-Gaussian norms.
\begin{defn}
(sub-exponential variable and sub-exponential norm) A random variable $X$ is called sub-exponential if there exists some positive constant $K_1$ such that $\mathbb{P}(|X|> t)\leqslant \exp(1-t/K_1)$ for all $t\geqslant 0$. The sub-exponential norm of $X$ is defined as $\|X\|_{\psi_1}=\sup_{q\geqslant 1}q^{-1}(\mathbb{E}|X|^q)^{1/q}$.
\end{defn}
\begin{defn}\label{defnSubGaussian}
(sub-Gaussian variable and sub-Gaussian norm) A random variable $X$ is called sub-Gaussian if there exists some positive constant $K_2$ such that $\mathbb{P}(|X|> t)\leqslant \exp(1-t^2/K_2)$ for all $t\geqslant 0$. The sub-Gaussian norm of $X$ is defined as $\|X\|_{\psi_2}=\sup_{q\geqslant 1}q^{-1/2}(\mathbb{E}|X|^q)^{1/q}$.
\end{defn}
\begin{lemmmm}\label{lemma3Han}
Let $\bar m_n$ be a sequence with $\bar m_n=o(n)$ and $\bar m_n\rightarrow \infty$ for $n\rightarrow \infty$
\begin{align*}
\max_{M:|M|\leqslant \bar m_n}\|{\theta}^M\|_2^2\leqslant C^2\frac{\bar m_n\log p}{n\phi^2_{\min}(\bar d)}.
\end{align*}
\begin{proof}
See Appendix \ref{prooflemma3Han}.
\end{proof}
\end{lemmmm}
Now define $A_{\lambda_n,\xi}$ to be
\begin{align*}
A_{\lambda_n,\xi}=\bigg\{k:\lambda_n\frac{\sqrt{d_k}\hat {\beta}_{jk}}{{\|\hat{{\beta}}_{J_k}} \|_2}=\frac{1}{n}{\psi}_j^T({Y}(\xi)-{\varphi}\hat{{\beta}}),
\mbox{ with }j\in J_k\bigg\},
\end{align*}
which represents the set of active groups for the de-noised problem.
\begin{lemmmm}\label{lemma4Han}
If, for a fixed value of $\lambda_n$, the number of active variables of the de-noised estimators $\hat{ {\beta}}^{\lambda_n,\xi}$ is for every $0\leqslant \xi \leqslant 1$ bounded by $m'$, then
\begin{align*}
\|\hat {{\beta}}^{\lambda_n,0}-\hat {{\beta}}^\lambda_n\|_2^2\leqslant C \max_{M:|M|\leqslant m'}\|{\theta}^M\|^2_2.
\end{align*}
\end{lemmmm}
\begin{proof}
See Appendix \ref{prooflemma4Han}.
\end{proof}

The next lemma provides an asymptotic upper bound on the number of selected variables.
\begin{lemmmm}\label{lemma6Han}
For $\lambda_n\geqslant \sqrt{\frac{\log p}{n}}$, the maximal number of selected variables, $\sup_{0\leqslant \xi \leqslant 1}\sum_{k\in A_{\lambda,\xi}}d_k$, is bounded, with probability tending to 1 for $n\rightarrow \infty$, by
\begin{align*}
\sup_{0\leqslant \xi \leqslant 1}\sum_{k\in A_{\lambda,\xi}}d_k\leqslant C_1s_n\bar d\bar c.
\end{align*}
\end{lemmmm}
\begin{proof}
See Appendix \ref{prooflemma6Han}.
\end{proof}
Now combining Lemmas \ref{lemma3Han}, \ref{lemma4Han}, and \ref{lemma6Han}, we have
\begin{align*}
\|\hat {{\beta}}^{\lambda_n,0}-\hat {{\beta}}^{\lambda_n}\|_2^2 \leqslant C\frac{s\bar d\bar c^2\log p}{n\phi^2_{\min}(s\bar d\bar c^2)}.
\end{align*}
Combining this and Lemma \ref{boundGamma}, gives
\begin{align*}
\|\hat{{\beta}}^{\lambda_n}-{\beta}\|^2_2&\leqslant C\frac{s\bar d\bar c^2\log p}{n\phi^2_{\min}(s\bar d\bar c^2)}+ \frac{\lambda_n^2 \bar c^2 s\bar d}{n^2}\bigg/\bigg(\sqrt{\frac{\kappa_{\min}}{2}(1-\frac{4\bar d}{\log n})}-\sqrt{\frac{2\kappa_{\max}\bar d^2}{\log n}}\bigg)^2\\
& \leqslant C\frac{s\bar d\bar c^2\log p}{n}+ C\frac{ \bar c^2 s\bar d \log p}{n}\bigg/\bigg(\sqrt{\frac{\kappa_{\min}}{2}(1-\frac{4\bar d}{\log n})}-\sqrt{\frac{2\kappa_{\max}\bar d^2}{\log n}}\bigg)^2\\
& \lesssim \frac{ \bar c^2s\bar d\log p}{n},
\end{align*}
which completes the proof of Theorem \ref{thm:L2consistency_det}. 

\subsection{Proof of Corollary \ref{Coro421}}
Since $\beta^*$ satisfies \eqref{betastardef},
\begin{align*}
\int_{\Omega} \varphi (x)(y(x) - \varphi (x)^T \beta^*) \mathrm{d} x = 0.
\end{align*}
Therefore, the oracle risk of $\hat \beta$ can be bounded by
\begin{align*}
& \int_{\Omega} (y(x) - \varphi (x)^T \hat{\beta})^2 \mathrm{d} x  - \int_{\Omega} (y(x) - \varphi (x)^T \beta^*)^2 \mathrm{d} x \\
 = &\int_{\Omega} (2y(x) - \varphi (x)^T \hat{\beta} - \varphi (x)^T \beta^*)(\varphi (x)^T (\beta^* - \hat{\beta})) \mathrm{d} x\\
 = & \int_{\Omega} (2y(x) - 2\varphi (x)^T \beta^* + \varphi (x)^T \beta^* - \varphi (x)^T \hat{\beta})(\varphi (x)^T (\beta^* - \hat{\beta})) \mathrm{d} x\\
 = & \int_{\Omega} (\varphi (x)^T \beta^* - \varphi (x)^T \hat{\beta})(\varphi (x)^T (\beta^* - \hat{\beta})) \mathrm{d} x \\
 = & \int_{\Omega} (\beta^* - \hat{\beta})^T \varphi (x)\varphi (x)^T (\beta^* - \hat{\beta}) \mathrm{d} x\\
 \leqslant & C\|\beta^* - \hat{\beta}\|_2^2,
\end{align*}
where the last inequality is because of Assumption \ref{mSparse}. Because $\|y(\cdot) - \varphi (\cdot)^T \beta^*\|_\infty = O_p(\lambda_n)$, we have $\int_{\Omega} (y(x) - \varphi (x)^T \beta^*)^2 \mathrm{d} x = O_p(\lambda_n^2)$, which completes the proof.


\subsection{Proof of Theorem \ref{thm:MainThmLinear}}\label{ProofThmmainLinear}
In this section we will prove Theorem \ref{thm:MainThmLinear}. A sketch of proof is as follows, following the overall approach in \cite{ning2014general}. First, we introduce a decorrelated score function, and prove the decorrelated function converges weakly to a normal distribution under $l_2$-consistency, which is stated in Theorem \ref{GeneralThm}. 
The result is then applied to the overlapping group lasso model with known variance of error.
Then by showing the difference between the decorrelated score function with known variance and decorrelated score function with estimated variance is small, we finish the proof of Theorem \ref{thm:MainThmLinear}.

\subsubsection{Hypothesis Test based on Decorrelated Function and $l_2$-Consistency}
In this section, we will introduce a decorrelated score function, and prove several results similar to \cite{ning2014general} but with $l_2$-consistency instead of $l_1$. Suppose we are given $n$ independently identically distributed $U_1,...,U_n$, which come from the same probability distribution following from a high dimensional statistical model $\mathcal{P}=\{\mathbb{P}_{{\beta}}:{\beta}\in \Omega\}$, where ${\beta}$ is a $p$ dimensional unknown parameter and $\Omega$ is the parameter space. Let the true value of ${\beta}$ be ${\beta}^*$, which is  sparse in the sense that the number of non-zero elements of $\beta$ is much smaller than $n$, order $\log n$. 
We consider the case in which we are interested in only one parameter. Suppose ${\beta}=({\beta_1},{\beta_{-1}})$, where ${\beta_1}\in\mathbb{R}$ and ${\beta_{-1}}\in\mathbb{R}^{p-1}$. Let $\beta_1^*$ and $\beta_{-1}^*$ be the true value of ${\beta_1}$ and $\beta_{-1}$, respectively. For simplicity, we assume the null hypothesis is $H_0:\beta_1^*=0$, which can be generalized to the case $\beta_1^*=\beta_{1,0}$ in a straight forward manner. Suppose the negative log-likelihood function is
\begin{align*}
\ell ({\beta_1},{\beta_{-1}})=\frac{1}{n}\sum_{i=1}^n (-\log f(U_i;{\beta_1},{\beta_{-1}})),
\end{align*}
where $f$ is the p.d.f. corresponding to the model $\mathbb{P}_{{\beta}}$, which it will be assumed has at least two continuous derivatives with respect to ${\beta}$. The information matrix for ${\beta}$ is defined as ${I}=\mathbb{E}_{{\beta}}(\nabla^2\ell({\beta}))$, and the partial information matrix is ${I}_{{\beta_1}|{\beta_{-1}}}=I_{{\beta_1}{\beta_1}}-{I}_{{\beta_1}{\beta_{-1}}}{I}_{{\beta_{-1}}{\beta_{-1}}}^{-1}{I}_{{\beta_{-1}}{\beta_1}}$, where $I_{{\beta_1}{\beta_1}}$, ${I}_{{\beta_1}{\beta_{-1}}}$, ${I}_{{\beta_{-1}}{\beta_{-1}}}$, and ${I}_{{\beta_{-1}}{\beta_1}}$ are the corresponding partitions of ${I}$. Let ${I}^*=\mathbb{E}_{{\beta}^*}(\nabla^2\ell({\beta}^*))$.\par
In this paper, we are considering testing parameters for high dimensional models and, as mentioned in \cite{ning2014general}, the traditional score function does not have a simple limiting distribution in the high dimensional setting. Thus, we use a  decorrelated score function as mentioned in \cite{ning2014general} defined as
\begin{align*}
S({\beta_1},{\beta_{-1}})=\nabla_{{\beta_1}} \ell({\beta_1},{\beta_{-1}})-{w}^T \nabla_{{\beta_{-1}}} \ell({\beta_1},{\beta_{-1}}),
\end{align*}
where ${w}={I}_{{\beta_{-1}}{\beta_{-1}}}^{-1}{I}_{{\beta_{-1}}{\beta_1}}$. Notice that $\mathbb{E}_{{\beta}}(S({\beta})\nabla_{\beta_{-1}} \ell({\beta}))=0$. Suppose we are given the estimator $\hat{{\beta}}=(\hat \beta_1,\hat \beta_{-1})$ and tuning parameter $\lambda'$. 
We estimate $\hat {{w}}$ by solving
\begin{align}\label{CalculateW}
\hat {{w}} =\arg\min\|{w}\|_1,\mbox{ s.t. }\|\nabla^2_{{\beta_1}{\beta_{-1}}} \ell(\hat{{\beta}})-{w}^T\nabla^2_{{\beta_{-1}}{\beta_{-1}}} \ell(\hat{{\beta}})\|_2\leqslant \lambda'.
\end{align}
We use this method to estimate ${w}$ because since ${w}$ has dimension $d$ which is much greater than $n$, we need some sparsity of ${w}$, which is useful in the rest part of this paper.
Thus, we can obtain estimated decorrelated score function $\hat S({\beta_1},\hat\beta_{-1})=\nabla_{\beta_1} \ell({\beta_1},\hat \beta_{-1})-\hat {{w}}^T \nabla_{{\beta_{-1}}} \ell({\beta_1},\hat \beta_{-1})$.\par
Along the same lines as \cite{ning2014general}, we need the following assumptions.
Assumption \ref{EstimationErrorBound} 
states that the estimators $\hat{{\beta}}$ and $\hat {{w}}$ 
converge to zero. However, we assume $l_2$-consistency here, which is weaker than the condition in \cite{ning2014general}. 
\begin{assum}\label{EstimationErrorBound}
Assume that
\begin{align*}
\lim_{n\rightarrow \infty} \mathbb{P}_{{\beta}^*}(\|\hat\beta_{-1}-\beta_{-1}^*\|_2\lesssim \eta_1(n))=1 \mbox{ and} & \lim_{n\rightarrow \infty} \mathbb{P}_{{\beta}^*}(\|\hat{{w}}-{w}^*\|_1\lesssim \eta_2(n))=1,
\end{align*}
where ${w}^*={I}^{*-1}_{{\beta_{-1}}{\beta_{-1}}}{I}^*_{{\beta_{-1}}{\beta_1}}$, and $\eta_1(n)$ and $\eta_2(n)$ converges to $0$, as $n\rightarrow \infty$.
\end{assum}
Assumption \ref{Noise} states that the derivative of log-likelihood function is near zero at the true parameters.
\begin{assum}\label{Noise}
Assume that
\begin{align*}
\lim_{n\rightarrow \infty} \mathbb{P}_{{\beta}^*}(\|\nabla_{{\beta_{-1}}} l(0,\beta_{-1}^*)\|_\infty\lesssim \eta_3(n))=1,
\end{align*}
for some $\eta_3(n)\rightarrow 0$, as $n\rightarrow \infty$.
\end{assum}
Assumption \ref{Stability} states that the Hessian matrix is relative smooth, so that we can use $\lambda'$ to control $\eta_4(n)$. 
\begin{assum}\label{Stability}
Assume that for ${\beta_{-1,\nu}} = \nu\beta_{-1}^*+(1-\nu)\hat\beta_{-1}$ with $\nu \in [0,1]$,
\begin{align*}
\lim_{n\rightarrow \infty} \mathbb{P}_{{\beta}^*}(\sup_{\nu \in [0,1]}\|\nabla^2_{{\beta_1}{\beta_{-1}}} l(0,\beta_{-1,\nu})-\hat{{w}}^T\nabla^2_{{\beta_{-1}}{\beta_{-1}}} l(0,\beta_{-1,\nu})\|_2\lesssim \eta_4(n))=1,
\end{align*}
for some $\eta_4(n)\rightarrow 0$, as $n\rightarrow \infty$.
\end{assum}
Assumption \ref{CLT} is the central limit theorem for a linear combination of the score functions. 
\begin{assum}\label{CLT}
For ${v}^*=(1,-{w}^{*T})^T$, it holds that
\begin{align*}
\frac{\sqrt{n}{v}^{*T}\nabla l(0,\beta_{-1}^*)}{\sqrt{{v}^T{I}^*{v}}}\stackrel{\rm dist.}{\longrightarrow} N(0,1),
\end{align*}
where $I^*=\mathbb{E}_{{\beta}^*}(\nabla^2 l(0,\beta_{-1}^*))$.
Furthermore, assume that $C'\leqslant I^*_{{\beta_1}|{\beta_{-1}}}<\infty$, where $I^*_{{\beta_1}|{\beta_{-1}}}=I^*_{{\beta_1}{\beta_1}}-{w}^{*T}I^*_{{\beta_{-1}}{\beta_1}}$, and $C'>0$ is a constant.
\end{assum}
Assumption \ref{HessianMatrix} states that we can estimate the information matrix relatively accurately.
\begin{assum}\label{HessianMatrix}
Assume
\begin{align*}
\lim_{n\rightarrow \infty}\mathbb{P}_{{\beta}^*}(\|\nabla^2 l(\hat{{\beta}})-{I}^*\|_{\max}\lesssim \eta_5(n))=1
\end{align*}
for some $\eta_5(n)\rightarrow 0$, as $n\rightarrow \infty$.
\end{assum}

Now under Assumptions \ref{EstimationErrorBound} to \ref{HessianMatrix}, we can prove a version of Theorem 3.5 in \cite{ning2014general} which applies to the (potentially) overlapping group lasso.
\begin{theo}\label{GeneralThm}
Under Assumptions \ref{EstimationErrorBound} to \ref{HessianMatrix}, with probability tending to one,
\begin{align}
n^{1/2}|\hat{S}(0,\hat\beta_{-1})-S(0,\beta_{-1}^*)|\lesssim n^{1/2}(\eta_2(n)\eta_3(n)+\eta_1(n)\eta_4(n)).\label{eq:generalBound}
\end{align}
If $n^{1/2}(\eta_2(n)\eta_3(n)+\eta_1(n)\eta_4(n))=o(1)$, we have
\begin{align}
n^{1/2}\hat{S}(0,\hat\beta_{-1})I_{{\beta_1}|{\beta_{-1}}}^{*-1/2}\stackrel{\rm dist.}{\longrightarrow} N(0,1).\label{eq:generalCLT}
\end{align}
\end{theo}
\begin{proof}
See Theorem 3.5 in \cite{ning2014general}. The only difference is under $l_2$-consistency,
\begin{align*}
|I_1|\leqslant \|\nabla^2_{{\beta_1}{\beta_{-1}}} l(0,\tilde\beta_{-1})-\hat{{w}}^T\nabla^2_{{\beta_{-1}}{\beta_{-1}}} l(0,\tilde\beta_{-1})\|_2\|\hat\beta_{-1} - \beta_{-1}^*\|_2\lesssim \eta_1(n)\eta_4(n).
\end{align*}
\end{proof}
\begin{coroll}\label{generalCorollary}
Assume that Assumptions \ref{EstimationErrorBound} to \ref{HessianMatrix} hold. It also holds that $\|{w}^*\|_1\eta_5(n)=o(1)$, $\eta_2(n)\|I^*_{{\beta_1}{\beta_{-1}}}\|_\infty=o(1)$, and $n^{1/2}(\eta_2(n)\eta_3(n)+\eta_1(n)\eta_4(n))=o(1)$. Under $H_0:\beta_1^*=0$, we have for any $t\in \mathbb{R}$,
\begin{align}\label{generalApproxNormal}
\lim_{n\rightarrow \infty}|\mathbb{P}_{{\beta}^*}(\hat{U}_n\leqslant t)-\Phi(t)|=0,
\end{align}
where $\hat{U}=n^{1/2}\hat S(0,\hat\beta_{-1})\hat I_{{\beta_1}|{\beta_{-1}}}^{-1/2}$.
\end{coroll}
\begin{proof}
See the proof of Corollary 3.7 in \cite{ning2014general}.
\end{proof}
\subsubsection{Linear model and the corresponding decorrelated score function}
Now we apply the consequences of the general results to the linear model as described in the previous section. In this section we first assume that the variance of noise is known. Consider the linear regression, $y_i=\beta_1^*\varphi_{i1}+\beta_{-1}^{*T}\varphi_{i,-1}+B_i + \epsilon_i$, where $\varphi_{i1}\in \mathbb{R}$, $\varphi_{i,-1}\in \mathbb{R}^{p-1}$, $B_i \in \mathbb{R}$, and the error $\epsilon_i$ satisfies $\mathbb{E}(\epsilon_i)=0$, $\mathbb{E}(\epsilon_i^2)=\sigma^2 > 0$ for $i=1,...,n.$ Let ${\varphi}_i=(\varphi_{i1},\varphi_{i,-1}^T)^T$ denote the collection of all covariates for subject $i$. We first assume $\sigma^2$ is known. \par
Consider the overlapping group lasso estimator (\ref{groupLassoEstimator1}), the decorrelated score function is
\begin{align*}
S({\beta_1},{\beta_{-1}})=-\frac{1}{n\sigma^2}\sum_{i=1}^n(y_i-{\beta_1} \varphi_{i1}-\beta_{-1}^T\varphi_{i,-1})(\varphi_{i1}-{w}^T\varphi_{i,-1}),
\end{align*}
where ${w}=\mathbb{E}_{{\beta}} (\varphi_{i,-1}\varphi_{i,-1}^T)^{-1}\mathbb{E}_{{\beta}} (\varphi_{i1}\varphi_{i,-1})$. Since the distribution of the design matrix does not depend on ${\beta}$, we can replace $\mathbb{E}_{{\beta}}(\cdot)$ by $\mathbb{E}(\cdot)$ for notation simplicity. Under the null hypothesis, $H_0:\beta_1^*=0$, the decorrelated score function can be estimated by
\begin{align*}
\hat S(0,\hat\beta_{-1})=-\frac{1}{n\sigma^2}\sum_{i=1}^n(y_i-\hat\beta_{-1}^T\varphi_{i,-1})(\varphi_{i1}-\hat w^T\varphi_{i,-1}),
\end{align*}
where
\begin{align*}
\hat {{w}} =\arg\min\|{w}\|_1, \mbox{ s.t. } \bigg\|\frac{1}{n}\sum_{i=1}^n \varphi_{i,-1}(\varphi_{i1}-{w}^T\varphi_{i,-1})\bigg\|_2\leqslant \lambda'.
\end{align*}
The (partial) information matrices are
\begin{align*}
{I}^*=\sigma^{-2}\mathbb{E}(\varphi_{i,-1}\varphi_{i,-1}^T), \mbox{ and } I^*_{{\beta_1}|{\beta_{-1}}}=\sigma^{-2}(\mathbb{E}(\varphi_{i1}^2)-\mathbb{E}(\varphi_{i1}\varphi_{i,-1}^T)\mathbb{E}(\varphi_{i,-1}\varphi_{i,-1}^T)^{-1}\mathbb{E}(\varphi_{i,-1}\varphi_{i1})),
\end{align*}
which can be estimated by
\begin{align*}
\hat {{I}}=\frac{1}{n\sigma^2}\sum_{i=1}^n \varphi_{i,-1}\varphi_{i,-1}^T, \mbox{ and } \hat I_{{\beta_1}|{\beta_{-1}}}=\sigma^{-2}\bigg\{\frac{1}{n}\sum_{i=1}^n \varphi_{i1}^2-\hat {{w}}^T\bigg(\frac{1}{n}\sum_{i=1}^n \varphi_{i,-1}\varphi_{i1}\bigg)\bigg\},
\end{align*}
respectively. Thus, the score test statistic is $\hat U_n = n^{1/2}\hat S(0,\hat\beta_{-1})\hat I_{{\beta_1}|{\beta_{-1}}}^{-1/2}$.\par
The following theorem states the asymptotic distribution $\hat U_n$ under null hypothesis.
\begin{theo}\label{MainThmLinear}
Assume that
\begin{enumerate}
    \item $\lambda_{\min}(\mathbb{E}( \varphi_{i}\varphi_{i}^T))\geqslant 2\kappa_{\min}$ for some constant $\kappa_{\min}>0$, and $\lim\sup_{n\rightarrow \infty} \phi_{\max}\leqslant \kappa_{\max}$, where $\phi_{\max}$ is defined in Definition \ref{msparseDef}.
    \item Let $S=\mbox{supp}({\beta}^*)$ and $S'=\mbox{supp}({w}^*)$ satisfy $|S|=s$ and $|S'|=s'$. Let $\bar{c}$ be the maximal number of replicates, $\bar d$ be the maximal number of group size. Assume $n^{-1/2}(s\vee s^*)\log p=o(1)$, $\bar d^2=o(\log n)$ and $\frac{\bar c^2\bar d}{\log p}=o(1)$.
    \item $\epsilon_i$, ${w}^{*T}\varphi_{i,-1}$, and $\varphi_{ij}$ are all sub-Gaussian with $\|\epsilon_i\|_{\Psi_2}\leqslant C$, $\|{w}^{*T}\varphi_{i,-1}\|_{\Psi_2}\leqslant C$, and $\|\varphi_{ij}\|_{\Psi_2}\leqslant C$, where $C$ is a positive constant.
    \item $\lambda'\asymp \sqrt{\frac{\log p}{n}}$ and $\lambda\asymp \sigma\sqrt{\frac{\log p}{n}}$.
    \item $B_i\lesssim \sqrt{\frac{\log p}{n}}$.
\end{enumerate}
Then under $H_0:\beta_1^*=0$ for each $t\in \mathbb{R}$,
\begin{align*}
\lim_{n\rightarrow \infty}|\mathbb{P}_{{\beta}^*}(\hat U_n\leqslant t)-\Phi(t)|=0.
\end{align*}
\end{theo}
\begin{proof}
Before the proof, we need the following lemmas in \cite{ning2014general}, which is used to ensure the assumptions of Theorem \ref{GeneralThm} and Corollary \ref{generalCorollary} hold. The proofs of Lemmas \ref{wstarbound}, \ref{wBound}, and \ref{supNormalBound} can be found in \cite{ning2014general}. In the proof of Lemma \ref{supNormalBound}, one need to notice that $\varphi^TB$ can be bounded by assumption. 
\begin{lemmmm}\label{wstarbound}
Under the conditions of Theorem \ref{MainThmLinear}, with probability at least $1-p^{-1}$, $\|\frac{1}{n}\sum_{i=1}^n(\varphi_{i1}\varphi_{i,-1}-\hat {{w}}^T\varphi_{i,-1}\varphi_{i,-1}^T)\|_\infty\leqslant C\sqrt{\frac{\log p}{n}}$, for some $C>0$.
\end{lemmmm}
\begin{lemmmm}\label{betaBound}
Under the conditions of Theorem \ref{MainThmLinear}, with probability at least $1-p^{-1}$,
\begin{align*}
\|\hat {{\beta}}-{\beta}^*\|^2_2\leqslant C_1\frac{ \bar c^2 s\bar d \log p}{n}, \mbox{ and } (\hat {{\beta}}-{\beta}^*)^T{H}_{{\varphi}}(\hat {{\beta}}-{\beta}^*)\leqslant C_1\kappa_{\max}\frac{ \bar c^2 s\bar d \log p}{n},
\end{align*}
where ${H}_{{\varphi}} = n^{-1}\sum_{i=1}^n {\varphi}_i{\varphi}_i^T$ and the constant $C_1>0$.
\end{lemmmm}
\begin{proof}
The first inequality is by Theorem \ref{thm:L2consistency_det}. The second inequality is trivial.
\end{proof}
\begin{lemmmm}\label{wBound}
Under the conditions of Theorem \ref{MainThmLinear}, with probability at least $1-p^{-1}$,
\begin{align*}
\|\hat {{w}}-{w}^*\|_1\leqslant 8C\kappa^{-1}s'\sqrt{\frac{\log p}{n}},
\end{align*}
where $C>0$ is a constant.
\end{lemmmm}
\begin{lemmmm}\label{supNormalBound}
Under the conditions of Theorem \ref{MainThmLinear}, it holds that $T^*\stackrel{\rm dist.}{\longrightarrow} N(0,1)$, and
\begin{align*}
\sup_{x\in\mathbb{R}}|\mathbb{P}_{{\beta}^*}(T^*\leqslant x)-\Phi(x)|\leqslant Cn^{-1/2},
\end{align*}
where $T^*=n^{1/2}S(0,\beta_{-1}^*)/I^{*1/2}_{{\beta_1}|{\beta_{-1}}}$ and $C$ is a positive constant not depending on ${\beta}^*$.
\end{lemmmm}
Now we can check that the assumptions of Theorem \ref{GeneralThm} and Corollary \ref{generalCorollary} hold, which finishes the proof of Theorem \ref{MainThmLinear}.
\end{proof}
Next we introduce some lemmas which give properties of sub-exponential variables and norms, as well as sub-Gaussian variables and norms, which will be used in the proof of Theorem \ref{thm:MainThmLinear}.
\begin{lemmmm}\label{Bernstein}
(Bernstein Inequality) Let $X_1,...,X_n$ be independent mean $0$ sub-exponential random variables and let $K=\max_i\|X_i\|_{\Psi_1}$. Then for any $t>0$,
\begin{align*}
\mathbb{P}_{\beta^*}\bigg(\frac{1}{n}\bigg|\sum_{i=1}^nX_i\bigg|\geqslant t\bigg)\leqslant 2\exp \bigg[-C\min\bigg(\frac{t^2}{K^2},\frac{t}{K}\bigg)n\bigg],
\end{align*}
where $C>0$ is a constant.
\end{lemmmm}
\begin{lemmmm}\label{XiEpsiloni}
Under the conditions of Theorem \ref{MainThmLinear} with probability at least $1-p^{-1}$, $\|\frac{1}{n}\sum_{i=1}^n\varphi_i\epsilon_i\|_\infty\leqslant C\sqrt{\frac{\log p}{n}}$, for some $C>0$.
\end{lemmmm}
The proofs of Lemmas \ref{Bernstein} and \ref{XiEpsiloni} can be found in \cite{ning2014general}.
Now, we can begin the proof of Theorem \ref{thm:MainThmLinear}.
\begin{proof}
The proof is similar to \cite{ning2014general} with a few changes. It is enough to show for any $\epsilon>0$,
\begin{align}\label{UtassymSame}
\lim_{n\rightarrow \infty} \sup_{{\beta}^*\in\Omega_0} \mathbb{P}_{{\beta}^*}(|\tilde U_n - \hat U_n|\geqslant \epsilon)=0.
\end{align}
Notice that $|\tilde U_n - \hat U_n|=|\hat U_n||1-\frac{\sigma^*}{\hat\sigma}|$. For a sequence of positive constants $t_n\rightarrow 0$ to be chosen later, we can show that $\lim_{n\rightarrow \infty}\sup_{\beta^*\in\Omega_0} \mathbb{P}_{{\beta}^*}(|\hat U_n|\geqslant t_n^{-1})=0$. It remains to show that
\begin{align}\label{sigmatn}
\lim_{n\rightarrow \infty}\sup_{{\beta}^*\in\Omega_0} \mathbb{P}_{{\beta}^*}\bigg(|1-\frac{\sigma^*}{\hat\sigma}|\geqslant t_n\bigg)=0.
\end{align}
Notice that
\begin{align}\label{sigmaBound}
\hat \sigma^2-\sigma^{*2} & =\bigg(\frac{1}{n}\sum_{i=1}^n(B_i+\epsilon_i)^2-\sigma^{*2}\bigg)+\hat {{\Delta}}^T{H}_{{\varphi}}\hat {{\Delta}}-2\hat {{\Delta}}^T\frac{1}{n}\sum_{i=1}^n(\epsilon_i + B_i){\varphi}_i\nonumber\\
& = \bigg(\frac{1}{n}\sum_{i=1}^n(B_i+\epsilon_i)^2-\sigma^{*2}\bigg)+\hat {{\Delta}}^T{H}_{{\varphi}}\hat {{\Delta}}-2\hat {{\Delta}}^T\frac{1}{n}\sum_{i=1}^n\epsilon_i {\varphi}_i-2\hat {{\Delta}}^T\frac{1}{n}\sum_{i=1}^nB_i {\varphi}_i\nonumber\\
&  = \bigg(\frac{1}{n}\sum_{i=1}^n\epsilon_i^2-\sigma^{*2}\bigg)+\hat {{\Delta}}^T{H}_{{\varphi}}\hat {{\Delta}}-2\hat {{\Delta}}^T\frac{1}{n}\sum_{i=1}^n\epsilon_i {\varphi}_i + \frac{1}{n}\sum_{i=1}^nB_i^2 + \frac{1}{n}\sum_{i=1}^n\epsilon_iB_i -2\hat {{\Delta}}^T\frac{1}{n}\sum_{i=1}^nB_i {\varphi}_i.
\end{align}
where $\hat {{\Delta}}=\hat {{\beta}}-{\beta}^*$. Since $\|\epsilon_i^2\|_{\psi_1}\leqslant 2C^2$, by Lemma \ref{Bernstein}, $|\frac{1}{n}\sum_{i=1}^n\epsilon_i^2-\sigma^{*2}|\leqslant C\sqrt{\frac{\log n}{n}}$, for some constant $C$, with probability tending to one. By Lemma \ref{betaBound}, we have ${\Delta}^T{H}_{{\varphi}}{\Delta}\leqslant C_1\kappa_{\max}\frac{ \bar c^2 s\bar d \log p}{n}$, for some constant $C_1$, with probability tending to one. By Lemma \ref{lemma6Han} and Lemma \ref{betaBound}, we have
\begin{align*}
\|\hat {{\Delta}}\|_1 & \leqslant C_1s\bar d\bar c^2\|\hat {{\Delta}}\|_2\\
                      &  \leqslant C_2s\bar d\bar c^2\sqrt{\frac{ \bar c^2 s\bar d \log p}{n}},
\end{align*}
for some constant $C_2>0$.
By Lemma \ref{XiEpsiloni}, we have
\begin{align*}
\bigg\| \frac{1}{n}\sum_{i=1}^n\epsilon_i{\varphi}_i \bigg\|_\infty\leqslant C_3\sqrt {\frac{\log p}{n}}.
\end{align*}
By Lemma \ref{Bernstein}, $|\frac{1}{n}\sum_{i=1}^n\epsilon_iB_i|\lesssim \sqrt{1/n}$. By the assumptions of Theorem \ref{MainThmLinear}, $\frac{1}{n}\sum_{i=1}^n B_i^2  \lesssim \frac{\log p}{n}$.
Thus,
\begin{align*}
\bigg|\hat {{\Delta}}^T\frac{1}{n}\sum_{i=1}^n\epsilon_i{\varphi}_i\bigg| & \leqslant \|\hat {{\Delta}}\|_1\bigg\| \frac{1}{n}\sum_{i=1}^n\epsilon_i{\varphi}_i \bigg\|_\infty\\
& \leqslant C_4 s\bar d\bar c^2\sqrt{\bar c^2 s\bar d}\frac{\log p}{n},
\end{align*}
for some constant $C_4>0$. By assumption $B_i\lesssim \sqrt{\frac{\log p}{n}}$, 
\begin{align*}
\bigg|\hat {{\Delta}}^T\frac{1}{n}\sum_{i=1}^nB_i{\varphi}_i\bigg| & \leqslant \|\hat {{\Delta}}\|_1\bigg\| \frac{1}{n}\sum_{i=1}^nB_i{\varphi}_i \bigg\|_\infty\\
& \leqslant C_5 s\bar d\bar c^2\sqrt{\bar c^2 s\bar d}\frac{\log p}{n},
\end{align*}
for some constant $C_5>0$.
Thus, by (\ref{sigmaBound}), we have
\begin{align*}
|\hat \sigma^2-\sigma^{*2}|\leqslant C_0 \sqrt{\frac{\log n}{n}}\vee (\bar c^2 s\bar d)^{3/2}\frac{\log p}{n},
\end{align*}
for some constant $C_0$, with probability tending to one. Thus,
\begin{align*}
|1-\frac{\sigma^*}{\hat\sigma}|=\hat\sigma^{-2}|1+\frac{\sigma^*}{\hat\sigma}||\hat \sigma^2-\sigma^{*2}|\lesssim |\hat \sigma^2-\sigma^{*2}|\lesssim \sqrt{\frac{\log n}{n}}\vee (\bar c^2 s\bar d)^{3/2}\frac{\log p}{n},
\end{align*}
with probability tending to one, because $\sigma^{*2}>C^2$ and $\hat\sigma^2=\sigma^{*2}+o_{\mathbb{P}}(1)$. Thus, if we choose $t_n \gtrsim \sqrt{\frac{\log n}{n}}\vee (\bar c^2 s\bar d)^{3/2}\frac{\log p}{n}$, then (\ref{sigmatn}) holds and (\ref{UtassymSame}) holds. Then by Theorem \ref{MainThmLinear}, the result holds.
\end{proof}

\subsection{Proofs of Lemmas}\label{PfofLemmasC}
\subsubsection{Proof of Lemma \ref{boundGamma}}\label{proofboundGamma}
\begin{proof}
For simplicity, we use $\lambda$ instead of $\lambda_n$, ${\gamma}$ instead of ${\gamma}^\lambda$, and ${\gamma}^Z$ instead of ${\gamma}^{Z,\lambda}$ in Appendix \ref{PfofLemmasC}. In this proof we will use ${\gamma}_t$ instead of ${\gamma}_{J_t}$ for brevity. Let ${\gamma}^Z(T)$ be the vector with elements $\gamma^Z_{ij_{ik}}(T)=\gamma^Z_{ij_{ik}}I_{\{\beta_i^*\neq 0\}}$. Similarly, $\gamma^Z_{ij_{ik}}(T^c)=\gamma_{ij_{ik}}^ZI_{\{\beta_i^*=0\}}$. Thus, ${\gamma}^Z={\gamma}^Z(T)+{\gamma}^Z(T^c)$. Notice $\{\beta_i^*\neq 0\}=\{i\in J_t, \mbox{ for some }t\in T\}$. Since $f(0,0)=0$, and (\ref{OptGamma}) is a minimizing problem, we have $f({\gamma},{\gamma}^Z)\leqslant 0$. Since ${\gamma}^T{C}{\gamma}\geqslant 0$ for any ${\gamma}$, and $\|{\beta}^Z_t\|_2-\|{\gamma}^Z_t+{\beta}^Z_t\|_2\leqslant \|{\gamma}^Z_t\|_2$ for any $t\in T$, combining $f({\gamma},{\gamma}^Z)\leqslant 0$, we have $\sum_{t\in T^c}\sqrt{d_t} \|{\gamma}^Z_t\|_2\leqslant \sum_{t\in T}\sqrt{d_t} \|{\gamma}^Z_t\|_2$. Also, we have
\begin{align}\label{equation5inlemma2Han}
\sum_{t\in T}\sqrt{d_t} \|{\gamma}^Z_t\|_2\leqslant \sqrt{\sum_{t\in T}d_t} \|{\gamma}^Z(T)\|_2\leqslant \sqrt{s_n\bar d}\|{\gamma}^Z\|_2.
\end{align}
The first inequality is true because of Cauchy's inequality, and the second inequality is true because $\bar d=\max\{d_1,...,d_n\}$ and $s_n=|T|$.\par
For any $\beta^{\lambda}_{ij_{im_1}}$ and $\beta^{\lambda}_{ij_{im_2}}$, if they are both not zero, by KKT conditions, we have
\begin{align*}
-\frac{1}{n}{\psi}_i^T(y- {\varphi}{\beta})+\frac{\lambda \sqrt{d_{j_{im_1}}} \beta^{\lambda}_{ij_{im_1}}}{\|{\beta}_{J_{j_{im_1}}}\|_2}=0,\mbox{ and }
-\frac{1}{n}{\psi}_i^T(y- {\varphi}{\beta})+\frac{\lambda \sqrt{d_{j_{im_2}}} \beta^{\lambda}_{ij_{im_2}}}{\|{\beta}_{J_{j_{im_2}}}\|_2}=0,
\end{align*}
which indicates
\begin{align*}
\frac{\lambda \sqrt{d_{j_{im_1}}} \beta^{\lambda}_{ij_{im_1}}}{\|{\beta}_{J_{j_{im_1}}}\|_2}=\frac{\lambda \sqrt{d_{j_{im_2}}} \beta^{\lambda}_{ij_{im_2}}}{\|{\beta}_{J_{j_{im_2}}}\|_2}.
\end{align*}
Since $\lambda>0$, we have $\beta^{\lambda}_{ij_{im_1}}\beta^{\lambda}_{ij_{im_2}}\geqslant 0$. Notice if $\beta^{\lambda}_{ij_{im_1}}$ or $\beta^{\lambda}_{ij_{im_2}}$ is zero, $\beta^{\lambda}_{ij_{im_1}}\beta^{\lambda}_{ij_{im_2}}\geqslant 0$ still holds. Together with the constraints of optimization problem, we have $\gamma^{\lambda}_{ij_{im_1}}\gamma^{\lambda}_{ij_{im_2}}\geqslant 0$, which indicates $\|{\gamma}^Z\|_2\leqslant \|{\gamma}\|_2$. Thus, together with (\ref{equation5inlemma2Han}), we have
\begin{align}\label{equation6inlemma2Han}
\sum_{t=1}^{p_n} \sqrt{d_t}\|{\gamma}^Z_t\|_2\leqslant 2\sqrt{s_n\bar d}\|{\gamma}^Z\|_2\leqslant 2\sqrt{s_n\bar d}\|{\gamma}\|_2.
\end{align}
Since $f({\gamma},{\gamma}^Z)\leqslant 0$, and ignoring the non-negative term $\lambda\sum_{t\in T^c}\sqrt{d_t} \|{\gamma}^Z_t\|_2$, it follows that
\begin{align}\label{equation7inlemma2Han}
n{\gamma}^T{C}{\gamma} \leqslant \lambda \sqrt{s_n\bar d}\|{\gamma}^Z\|_2 \leqslant \lambda \sqrt{s_n\bar d}\|{\gamma}\|_2.
\end{align}
Next, we bound the term $n{\gamma}^T{C}{\gamma} $ from below. 
Pplugging the result into (\ref{equation7inlemma2Han}) will yield the desired upper bound on the $l_2$-norm of ${\gamma}$. Let $\|{\gamma}^Z_{(1)}\|_2\geqslant \|{\gamma}^Z_{(2)}\|_2\geqslant \cdots \geqslant\|{\gamma}^Z_{(p_n)}\|_2$ be the ordered block entries of ${\gamma}$. Let $\{u_n\}$ be a sequence of positive integers, such that $1\leqslant u_n\leqslant p_n$ and define the set of $u_n$-largest groups as $U=\{k:\|{\gamma}^Z_k\|_2\geqslant \|{\gamma}^Z_{(u_n)}\|_2\}$. Define analogously as before ${\gamma}^Z(U)$, ${\gamma}^Z(U^c)$, ${\gamma}(U)$, and ${\gamma}(U^c)$. Thus, ${\gamma}^T{C}{\gamma} = ({\gamma}(U)+{\gamma}(U^c))^T{C}({\gamma}(U)+{\gamma}(U^c)=\|{a}+{b}\|_2^2$, where ${a}={\varphi}{\gamma}(U)/\sqrt{n}$ and ${b}={\varphi}{\gamma}(U^c)/\sqrt{n}$. Thus,
\begin{align}\label{equation8inlemma2Han}
{\gamma}^T{C}{\gamma}={a}^T{a}+2{b}^T{a}+{b}^T{b}\geqslant (\|{a}\|_2-\|{b}\|_2)^2.
\end{align}
Assume $l=\sum_{t=1}^{p_n}\|{\gamma}^Z_t\|_2$. Then for every $t=1,...,p_n$, $\|{\gamma}^Z_{(t)}\|_2\leqslant l/t$, since ${\gamma}^Z_{(t)}$ is the $t^{th}$ largest group with respect to $\|\cdot\|_2$. Thus,
\begin{align}\label{gammaU}
\|{\gamma}^Z(U^c)\|^2_2= \sum_{t=u_n+1}^{p_n} \|{\gamma}^Z_{(t)}\|_2\leqslant \bigg(\sum_{t=1}^{p_n}\|{\gamma}^Z_t\|_2^2\bigg)^2\sum_{t=u_n+1}^{p_n}\frac{1}{t^2}\leqslant \bigg(\sum_{t=1}^{p_n}\sqrt{d_t}\|\gamma^Z_t\|_2\bigg)^2\frac{1}{u_n},
\end{align}
where the last inequality is because
\begin{align*}
\sum_{t=u_n+1}^{p_n}\frac{1}{t^2}\leqslant \int_{s=u_n}^\infty \frac{1}{s^2} ds = \frac{1}{u_n},
\end{align*}
and $\sqrt{d_t}\geqslant 1$.\par
Together with (\ref{equation6inlemma2Han}), we have $\|{\gamma}^Z(U^c)\|^2_2\leqslant 4s_n\bar d\|{\gamma}^Z\|_2^2\frac{1}{u_n}$. Since ${\gamma}(U)$ has at most $\sum_{t\in U}d_t$ non-zero coefficients, and $\sum_{t\in U}d_t\leqslant u_n\bar d$,
\begin{align}\label{equation9inlemma2Han}
\|{a}\|_2^2 & \geqslant \phi_{\min}\bigg(\sum_{t\in U}d_t\bigg)\|{\gamma}(U)\|_2^2
\geqslant \phi_{\min}\bigg(\sum_{t\in U}d_t\bigg)\|{\gamma}^Z(U)\|_2^2\nonumber\\
& = \phi_{\min}\bigg(\sum_{t\in U}d_t\bigg)(\|{\gamma}^Z\|_2^2-\|{\gamma}^Z(U^c)\|_2^2)
\geqslant \phi_{\min}\bigg(\sum_{t\in U}d_t\bigg)(1-\frac{4s_n\bar d}{u_n})\|{\gamma}^Z\|_2^2\nonumber\\
&\geqslant \phi_{\min}(u_n\bar d)(1-\frac{4s_n\bar d}{u_n})\|{\gamma}^Z\|_2^2.
\end{align}
The first inequality is true because of the definition of $\phi_{\min}(\cdot)$, and the equality is true because $\gamma^Z=\gamma^Z(U)+\gamma^Z(U^c)$.
From Lemma \ref{lemma1Han}, $\gamma(U^c)$ has at most $n$ non-zero groups, which indicates
\begin{align}\label{equation10inlemma2Han}
\|{b}\|_2^2\leqslant \phi_{\max}(n\bar d)\|{\gamma}(U^c)\|_2^2\leqslant \phi_{\max}\|{\gamma}(U^c)\|_2^2\leqslant \bar d\phi_{\max}\|{\gamma}^Z(U^c)\|_2^2\leqslant \frac{4\phi_{\max}s_n\bar d^2}{u_n}\|{\gamma}^Z\|_2^2.
\end{align}
The first inequality is true because the definition of $\phi_{\max}(\cdot)$, the third inequality is true is because of Cauchy's inequality, and the last inequality is true because of (\ref{equation6inlemma2Han}) and (\ref{gammaU}).
Thus, plugging (\ref{equation9inlemma2Han}) and (\ref{equation10inlemma2Han}) into (\ref{equation8inlemma2Han}), and combining with the facts $\sum_{t\in U} d_t\leqslant \bar du_n$ and $\phi_{\max}\geqslant \phi_{\min}(u_n)$, under Assumption \ref{mSparse}, for sufficient large $n$, we have
\begin{align*}
\|{a}\|_2-\|{b}\|_2 & \geqslant \bigg(\sqrt{\phi_{\min}(u_n\bar d)(1-\frac{4s_n\bar d}{u_n})}-\sqrt{\frac{4\phi_{\max}s_n\bar d^2}{u_n}}\bigg)\|{\gamma}^Z\|_2\\
& \geqslant \bigg(\sqrt{\phi_{\min}(u_n\bar d)(1-\frac{4s_n\bar d}{u_n})}-\sqrt{\frac{2\kappa_{\max}s_n\bar d^2}{u_n}}\bigg)\|{\gamma}^Z\|_2
\end{align*}
Let $u_n=s_n\log n$, under Assumption \ref{mSparse}, for large $n$, we have
\begin{align*}
\|{a}\|_2-\|{b}\|_2 & \geqslant \bigg(\sqrt{\frac{\kappa_{\min}}{2}(1-\frac{4\bar d}{\log n})}-\sqrt{\frac{2\kappa_{\max}\bar d^2}{\log n}}\bigg)\|{\gamma}^Z\|_2.
\end{align*}
Together with (\ref{equation7inlemma2Han}), we have
\begin{align*}
\frac{\lambda \sqrt{s_n\bar d}}{n}\|{\gamma}^Z\|_2 \geqslant {\gamma}^T{C}{\gamma} \geqslant
\bigg(\sqrt{\frac{\kappa_{\min}}{2}(1-\frac{4\bar d}{\log n})}-\sqrt{\frac{2\kappa_{\max}\bar d^2}{\log n}}\bigg)^2\|{\gamma}^Z\|^2_2.
\end{align*}
Since by Cauchy's inequality, we have $\|{\gamma}^Z\|^2_2\geqslant \|{\gamma}\|^2_2/\bar c$. Thus,
\begin{align*}
\|{\gamma}\|_2^2\leqslant \frac{\lambda^2 \bar c s_n\bar d}{n^2}\bigg/\bigg(\sqrt{\frac{\kappa_{\min}}{2}(1-\frac{4\bar d}{\log n})}-\sqrt{\frac{2\kappa_{\max}\bar d^2}{\log n}}\bigg)^2,
\end{align*}
which completes the proof.
\end{proof}
\subsubsection{Proof of Lemma \ref{lemma3Han}}\label{prooflemma3Han}
\begin{proof}
From (\ref{equation12inlemma2Han}), for every $M$ with $|M|\leqslant \bar m_n$,
\begin{align}\label{equation13inlemma2Han}
\|{\theta}^M\|_2^2\leqslant \frac{1}{n^2\phi^2_{\min}(\bar m_n)}\|{\varphi}_M^T(\epsilon+B)\|_2^2 \leqslant \frac{2}{n^2\phi^2_{\min}(\bar m_n)}(\|{\varphi}_M^T\epsilon\|_2^2 + \|{\varphi}_M^TB\|_2^2)
\end{align}

By Lemma \ref{XiEpsiloni}, with probability at least $1-d^{-1}$, $\|\sum_{i=1}^n\varphi_i\epsilon_i\|_\infty\leqslant C_1\sqrt{n\log p}$. Thus,
\begin{align*}
\max_{M:|M|\leqslant \bar m_n} \|{\varphi}^T_M{\epsilon}\|_2^2\leqslant \bar m_n \|\sum_{i=1}^n\varphi_i\epsilon_i\|_\infty^2 \leqslant \bar m_n C_1^2n\log p,
\end{align*}
where the first inequality is true because $\|{\varphi}^T_M{\epsilon}\|_2^2\leqslant |M|\|{\varphi}^T_M{\epsilon}\|_\infty^2$, and $|M|\leqslant \bar m_n$.

By assumptions of Theorem \ref{thm:L2consistency_det}, 
\begin{align*}
\max_{M:|M|\leqslant \bar m_n} \|{\varphi}^T_M{B}\|_2^2\leqslant \bar m_n \|\sum_{i=1}^n\varphi_iB_i\|_\infty^2 \leqslant \bar m_n C_2^2n\log p.
\end{align*}

Thus,
\begin{align*}
\max_{M:|M|\leqslant \bar m_n}\|{\theta}^M\|_2^2\leqslant C^2\frac{\bar m_n\log p}{n\phi^2_{\min}(\bar m_n)},
\end{align*}
which finishes the proof.
\end{proof}
\subsubsection{Proof of Lemma \ref{lemma4Han}}\label{prooflemma4Han}
\begin{proof}
Before the proof, we state a lemma.
\begin{lemmmm}\label{lemma5Han}
For ${x}\in \mathbb{R}^q$, suppose $\hat {{x}}_1 =\arg\min_{{x}}f_1({x})$ and $\hat {{x}}_2 =\arg\min_{{x}}f_2({x})$ where $f_1({x})=\frac{1}{2}{x}^T{A}^T{A}{x}+{b}^T{x}$ with ${A}\in \mathbb{R}^{n\times q}$ which is full rank and ${b}\in \mathbb{R}^q$. Also, $f_2({x})=f_1({x})+{c}^T{x}$ with ${c}\in \mathbb{R}^q$. Let ${A}^Z$, ${b}^Z$ and ${c}^Z$ be defined in the same way as before. Let $g_1({y}^Z)=\frac{1}{2}\|{A}^Z{y}^Z\|_2^2+({b}^Z)^T{y}^Z+h({y}^Z)$ and $g_2({y}^Z)=\frac{1}{2}\|{A}^Z{y}^Z\|_2^2+({b}^Z)^T{y}^Z+({c}^Z)^T{y}^Z+h({y}^Z)$, where $h({y})$ is a convex function with respect to ${y}$ and everywhere sub-differentiable, and define $\hat {{y}}^Z_1 =\arg\min_{{y}}^Zg_1({y}^Z)$ and $\hat {{y}}^Z_2 =\arg\min_{{y}}^Zg_1({y}^Z)$. Then we have
\begin{align*}
\|\hat {{y}}_2-\hat {{y}}_1\|_2\leqslant {\gamma} \|\hat {{x}}_2-\hat {{x}}_1\|_2.
\end{align*}
\end{lemmmm}
\begin{proof}
Our proof is similar to \cite{liu2009estimation}, with the only difference that $\|{A}^Z(\hat {{y}}^Z_1-\hat {{y}}^Z_2)\|_2^2+({c}^Z)^T(\hat {{y}}^Z_1-\hat {{y}}^Z_2)=\|{A}(\hat {{y}}_1-\hat {{y}}_2)\|_2^2+{c}^T(\hat {{y}}_1-\hat {{y}}_2)$.
\end{proof}
Let $M(\xi)=A_{\lambda,\xi}$. Let $0=\xi_1<...<\xi_{J+1}=1$ be the points of discontinuity of $M(\xi)$. At these locations, variables either join the active set or are dropped from the active set. Fix some $j$ with $1\leqslant j\leqslant J$. Denote by $M_j$ be the set of active groups $M(\xi)$ for any $\xi\in(\xi_j,\xi_{j+1})$. Assuming
\begin{align}\label{equation15Han}
\forall \xi\in(\xi_j,\xi_{j+1}):\|\hat {{\beta}}^{\lambda,\xi}-\hat {{\beta}}^{\lambda,\xi_j}\|_2\leqslant C(\xi-\xi_{j})\|\hat {{\theta}}^{M_j}\|_2
\end{align}
is true, where ${\theta}^{M_j}$ is the restricted OLS estimator of noise. Then
\begin{align*}
\|\hat {{\beta}}^{\lambda,0}-\hat {{\beta}}^\lambda\|_2 & \leqslant \sum_{j=1}^J\|\hat {{\beta}}^{\lambda,\xi_j}-\hat {{\beta}}^{\lambda,\xi_{j+1}}\|_2\\
                                                          & \leqslant C \max_{M:|M|\leqslant m}\|{\theta}^M\|_2 \sum_{j=1}^J (\xi_{j+1}-\xi_j)\\
                                                          & =C \max_{M:|M|\leqslant m}\|{\theta}^M\|_2.
\end{align*}
By replacing $\hat {{x}}_1$, $\hat {{x}}_2$, $\hat {{y}}_1$ and $\hat {{y}}_2$ with $\xi\hat {{\theta}}^{M_j}$, $\xi_j\hat {{\theta}}^{M_j}$, $\hat {{\beta}}^{\lambda,\xi}$ and $\hat {{\beta}}^{\lambda,\xi_j}$ in Lemma \ref{lemma5Han}, respectively, we obtain (\ref{equation15Han}). Hence, we complete the proof.
\end{proof}
\subsubsection{Proof of Lemma \ref{lemma6Han}}\label{prooflemma6Han}
\begin{proof}
Our proof is similar to \cite{meinshausen2009lasso}. The only thing need to be noticed is that for (38) in \cite{meinshausen2009lasso}, we have
\begin{align*}
(\|({\varphi}^Z_{A_{\lambda,\xi}})^T{\varphi}({\beta}-\hat {{\beta}}^{\lambda,\xi})\|_2+\|({\varphi}^Z_{A_{\lambda,\xi}})^T(\epsilon+B)\|_2)^2 & \leqslant 2(\|({\varphi}^Z_{A_{\lambda,\xi}})^T{\varphi}({\beta}-\hat {{\beta}}^{\lambda,\xi})\|^2_2+\|({\varphi}^Z_{A_{\lambda,\xi}})^T{(\epsilon+B)}\|^2_2)\\
& \leqslant 2\bar c (\|{\varphi}_{A_{\lambda,\xi}}^T{\varphi}({\beta}-\hat {{\beta}}^{\lambda,\xi})\|^2_2+\|{\varphi}_{A_{\lambda,\xi}}^T({\epsilon} +B)\|^2_2).
\end{align*}
\end{proof}

\subsection{Description of Functions in Section \ref{sec:morefunctions}}\label{append:morefunctions}
\begin{itemize}
\item The amount of deflection of a bending function is given by
\[
D_e=\frac{4}{10^9}\frac{L^3}{bh^3},
\]
where the 3 inputs are $L,b$, and $h$.
\item The midpoint voltage of a transformerless OTL circuit function is given by
\[
V_m=\frac{(V_{b1}+0.74)B(R_{c2}+9)}{B(R_{c2}+9)+R_f}+\frac{11.35R_f}{B(R_{c2}+9)+R_f}+\frac{0.74R_f\beta(R_{c2}+9)}{(B(R_{c2}+9)+R_f)R_{c1}},
\]
where $V_{b1}=12R_{b2}/(R_{b1}+R_{b2})$, and the 6 inputs are $R_{b1},R_{b2},R_{f},R_{c1},R_{c2},$ and $B$.

\item The wing weight function models a light aircraft wing, where the wing's weight is given by
\begin{equation}\label{eq:wingfunction}
W=0.036S^{0.758}_wW^{0.0035}_{fw}\left(\frac{A}{\cos^2(\Lambda)}\right)^{0.6}q^{0.006}R^{0.04}\left(\frac{100t_c}{\cos(\Lambda)}\right)^{-0.3}(N_zW_{dg})^{0.49}+S_wW_p,
\end{equation}
where the 10 inputs are $S_w,W_{fw},A,\Lambda,q,R,t_c,N_z,W_{dg},$ and $W_p$.
\end{itemize}

The input ranges are given in Table \ref{tab:inputrange}.
\begin{table}[h]
\centering
\begin{tabular}{rlrlrl}
\toprule
\multicolumn{2}{c}{Bending} & \multicolumn{2}{c}{OTL circuit} & \multicolumn{2}{c}{Wing weight}\\
\midrule
$L$ & $\in[10,20]$ & $R_{b1}$ & $\in[50,150]$  & $S_w$ & $\in[150,200]$\\
$b$ & $\in[1,2]$ & $R_{b2}$ & $\in[25,70]$ & $W_{fw}$ & $\in[220,300]$\\
$h$ & $\in[0.1,0.2]$ & $R_{f}$ & $\in[0.5,3]$& $A$ & $\in[6,10]$\\
& & $R_{c1}$ & $\in[1.2,2.5]$ & $\Lambda$ & $\in[-10,10]$\\
& & $R_{c2}$ & $\in[0.25,1.2]$ & $q$ & $\in[16,45]$\\
& & $\beta$ & $\in[50,300]$ & $R$ & $\in[0.5,1]$\\
& & & & $t_c$ & $\in[0.08,0.18]$\\
& & & & $N_z$ & $\in[2.5,6]$\\
& & & & $W_{dg}$ & $\in[1700,2500]$\\
& & & & $W_p$ & $\in[0.025,0.08]$\\
\bottomrule
\end{tabular}
\caption{Input ranges of the OTL circuit function, the piston simulation function, and the wing weight function.}
\label{tab:inputrange}
\end{table}

\newpage 

\bibliography{bib}

\end{document}